\documentclass[acmtog,screen]{acmart}

\AtBeginDocument{%
  }

\setcopyright{cc}
\setcctype{by}
\acmJournal{TOG}
\acmYear{2026}
\acmVolume{45}
\acmNumber{4}
\acmArticle{43}
\acmMonth{7}
\acmDOI{10.1145/3811323}

\PassOptionsToPackage{prologue,table,xcdraw}{xcolor}
\usepackage{xargs}
\usepackage{booktabs} 
\usepackage{tabularx}
\usepackage{caption}
\usepackage{nicematrix}
\usepackage{float}
\usepackage{mathtools}
\usepackage{amsmath,amsfonts}
\usepackage{bm}
\usepackage{microtype}
\usepackage{units}
\usepackage{subcaption}
\usepackage{overpic}
\usepackage{cancel}
\usepackage{wrapfig}
\usepackage{placeins}
\usepackage{empheq}
\usepackage[normalem]{ulem}
\usepackage[export]{adjustbox}
\usepackage[vlined,boxed,ruled]{algorithm2e} 
\usepackage{tikz}
\usepackage{xstring}
\usepackage{xr}
\usepackage{multirow}
\usepackage{makecell}
\usepackage[capitalize,noabbrev]{cleveref} 
\usepackage{url} 
\usepackage[utf8]{inputenc}
\usepackage{enumitem,enumerate}
\usepackage{wrapfig}

\setlist[itemize]{noitemsep, nolistsep, leftmargin=*}

\SetKw{Or}{or}
\SetKw{And}{and}
\SetStartEndCondition{ }{}{}%
\SetKwProg{Fn}{def}{\string:}{}
\SetKwFunction{Range}{range}
\SetKw{KwTo}{in}
\SetKwFor{For}{for}{\string:}{}%
\SetKwFor{While}{while}{:}{fintq}%
\SetKwIF{If}{ElseIf}{Else}{if}{:}{elif}{else:}{}%
\SetKwComment{Comment}{\# }{}
\SetKw{Continue}{continue}
\AlgoDontDisplayBlockMarkers
\SetAlgoNoEnd
\SetAlgoNoLine
\DontPrintSemicolon

\SetAlFnt{\small}
\SetAlCapFnt{\small}
\SetAlCapNameFnt{\small}
\SetAlCapHSkip{0pt}
\IncMargin{-\parindent}

\def\commentType{1}

\ifnum\commentType=0
    \newcommandx{\customComment}[3]{}
    \newcommandx{\customTODO}[3]{}
\fi
\ifnum\commentType=1
    \newcommandx{\customComment}[3]{\textcolor{#2}{\textsl{#1: #3}}}
    \newcommandx{\customTODO}[3]{\textcolor{#2}{\textsl{#1: #3}}}
\fi
\ifnum\commentType=2
    \usepackage{pdfcomment}
    \newcommandx{\customComment}[3]{\pdfcomment[icon=Comment,opacity=0.5,color=#2,author=#1]{#3}}
    \newcommandx{\customTODO}[3]{\pdfcomment[icon=Note,opacity=0.5,color=#2,author=#1]{#3}}
\fi
\ifnum\commentType=3
    
    \usepackage{pdfcomment} 
    \usepackage{todonotes} 
    \usepackage{silence}
    \newcommandx{\customComment}[3]{\todo[color=#2!40,size=\small]{\textbf{#1:} #3}}
    \newcommandx{\customTODO}[3]{\todo[color=#2!40,size=\small]{\textbf{#1:} #3}}
\fi

\let\originalleft\left 
\let\originalright\right 
\renewcommand{\left}{\mathopen{}\mathclose\bgroup\originalleft} 
\renewcommand{\right}{\aftergroup\egroup\originalright} 

\definecolor{amber}{rgb}{1.0, 0.49, 0.0}
\definecolor{darkgreen}{rgb}{0.0, 0.5, 0.0}
\definecolor{darkblue}{rgb}{0.0, 0.0, 0.5}
\definecolor{darkred}{rgb}{0.5, 0.0, 0.0}

\newcommandx{\All}[1]{\customComment{All}{red}{#1}}
\newcommandx{\Cedric}[1]{\customComment{C}{darkred}{#1}}
\newcommandx{\Philip}[1]{\customComment{P}{darkred}{#1}}
\newcommandx{\Misha}[1]{\customComment{M}{darkblue}{#1}}
\newcommandx{\TODO}[1]{\customTODO{TODO}{red}{#1}}
\newcommandx{\todo}[1]{\customTODO{TODO}{red}{#1}}

\newcommand{\REMOVE}[1]{} 

\def\equationautorefname~#1\null{%
  Equation~(#1)\null
}

\newcommand{\norm}[1]{\lVert #1 \rVert}
\renewcommand{\vec}[1]{\boldsymbol{#1}}  
\newcommand{\R}[0]{{\mathbb{R}}}
\newcommand{\Rtwo}[0]{{\mathbb{R}^2}}
\newcommand{\Rthree}[0]{{\mathbb{R}^3}}
\newcommand{\M}[0]{\mathcal{M}}

\newcommand{\Mesh}[0]{\mathcal{M}}
\newcommand{\MeshBoundary}[0]{\partial \mathcal{M}}

\DeclareMathAlphabet{\mathmybb}{U}{bbold}{m}{n}

\newcommand{\bigO}{\mathcal{O}}

\definecolor{cartoPrismTeal}{rgb}{0.21960784 0.65098039 0.64705882}
\definecolor{cartoPrismOrange}{rgb}{0.88235294 0.48627451 0.01960784}
\definecolor{cartoPrismGreen}{rgb}{0.45098039 0.68627451 0.28235294}
\definecolor{cartoPrismRed}{rgb}{0.8 0.31372549 0.24313725}
\definecolor{cartoPrismPurple}{rgb}{0.58039216 0.20392157 0.43137255}

\definecolor{mathematicaBlue}{rgb}{0.38, 0.51, 0.71}
\definecolor{mathematicaOrange}{rgb}{0.88, 0.61, 0.14}
\definecolor{mathematicaGreen}{rgb}{0.56, 0.69, 0.19}
\definecolor{mathematicaRed}{rgb}{0.92,0.39, 0.21}
\definecolor{mathematicaPurple}{rgb}{0.53, 0.47, 0.7}

\newcommand{\StwoP}{S^2_{\vec{p}}}
\newcommand{\T}{\mathcal{T}}
\newcommand{\South}{\vec{x}_0}
\newcommand{\North}{\vec{x}_1}

\newcommand{\SphericalArea}{ \textrm{Area}_{\StwoP}}
\newcommand{\EulerChar}{\mathcal{E}}
\newcommand{\Intersections}{\mathcal{I}}
\newcommand{\VF}{X}
\newcommand{\Weight}{\omega}
\newcommand{\Curve}{\Gamma}

\newif\ifshownew
\shownewtrue

\ifshownew
\newenvironment{isnew}{\color{darkgreen}}{}
\else

\fi

\DeclareMathOperator{\ind}{ind}

\newcommand\encircle[1]{%
	\tikz[baseline=(X.base)]
	\node (X) [draw, shape=circle, inner sep=0] {\strut #1};}

\definecolor{greenish}{RGB}{20,92,44}

\SetCommentSty{mycommfont}

\usepackage{chapterbib}

\definecolor{mygreen}{rgb}{0.019, 0.443,0.07}

\makeatletter
\patchcmd{\@mkbibcitation}{\ref{TotPages}}{11}{}{}
\patchcmd{\@mkbibcitation}{\getrefnumber{TotPages}}{11}{}{}
\makeatother

\newtheorem{theorem}{Theorem}[section]

\newtheorem{lemma}[theorem]{Lemma}
\theoremstyle{definition}

\newtheorem{proposition}[theorem]{Proposition}

\AtBeginDocument{\hypersetup{colorlinks=true, linkcolor=blue, citecolor=blue, urlcolor=blue}}

\begin{document}


\title{The Antipodal Method: Fast, Accurate, and Robust 3D Generalized Winding Numbers}

\author{Cedric Martens}
\authornote{Equal Contribution}
\email{cedric.martens@umontreal.ca}
\affiliation{
  \institution{Université de Montréal}
  \city{Montréal}
  \country{Canada}
}

\author{Philip Trettner}
\authornotemark[1]
\email{trettner@shapedcode.com}
\affiliation{
  \institution{Shaped Code GmbH}
  \city{Aachen}
  \country{Germany}
}

\author{Mikhail Bessmeltsev}
\email{bmpix@iro.umontreal.ca}
\affiliation{
  \institution{Université de Montréal}
  \city{Montréal}
  \country{Canada}
}

\begin{abstract}
Generalized winding numbers provide a robust measure of point insidedness for 3D surfaces—whether open, self-intersecting, or non-manifold—and are central to numerous geometry processing tasks. However, existing methods trade off between accuracy and computational efficiency, limiting their use in interactive and large-scale applications.

We introduce a new formulation and algorithm for computing generalized winding numbers that is both fast and accurate to arbitrary precision, applicable to meshes and parametric surfaces. Our approach expresses the winding number as the sum of two intuitive geometric quantities: the signed number of ray–surface intersections and a boundary integral over the surface’s projection onto the unit sphere. This insight leads to an efficient discretization that avoids expensive surface integrals and spherical arrangements.

For meshes, our method achieves average speedups of $22\times$ on a CPU compared to the fastest precise methods and $3\times$ compared to the fastest approximation method, while maintaining full precision. On a GPU, for moderately complex meshes we reach a throughput of $10^9$ queries per second, or $4K$ generalized winding number slices at 120 FPS ($13\times$ faster than a na\"ive GPU method). For parametric surfaces, our method is on average $5.6\times$ faster than the state-of-the-art method, with the same precision. Our method naturally handles complex topologies and non-manifold inputs. We extensively validate its accuracy, robustness, and time performance. Our code is available at \url{https://github.com/MartensCedric/antipodal}.
\end{abstract}



\begin{teaserfigure}
	\includegraphics[width=\textwidth]{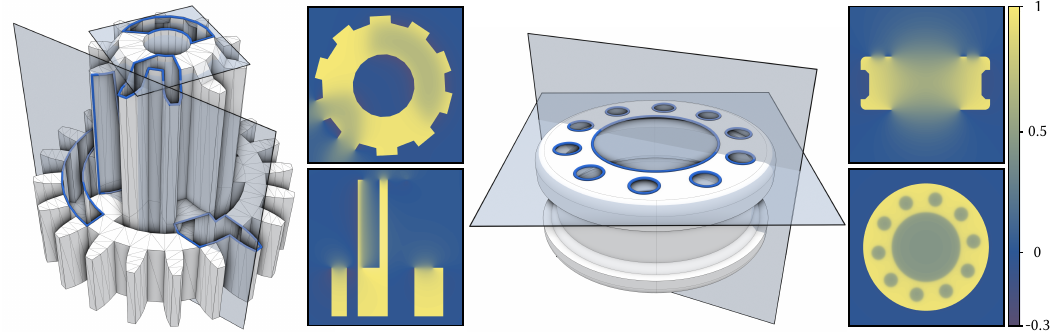}
	\caption{We propose a new method for generalized winding numbers in 3D, computing them as a combination of a ray-surface intersection and a line integral over the surface boundary projected onto a unit sphere. Our method is significantly faster than the state-of-the-art methods for both meshes (left) and parametric surfaces (right), while staying accurate to arbitrary precision and robust to degenerate configurations.}
	\Description{figure description}
	\label{fig:teaser}
\end{teaserfigure}

\maketitle

\section{Introduction}
Generalized winding numbers (GWNs), quantifying \emph{how much} a given point is inside a 2D curve or a 3D surface, be it open, self-intersecting, or even non-manifold, are a standard tool in robust geometry processing. GWNs allow to elegantly handle imperfections of real-world data and thus have found numerous applications ranging from tetrahedral meshing  \cite{ichim2017phace, tetrahedralMeshing2018} and shape generation \cite{metzer2023latent} to mesh booleans \cite{cherchi2022interactive, Trettner2022}, surface reconstruction and normal estimation \cite{wang2022restricted, he2024windpoly,metzer2021orienting,OrientingPointclouds2023}, and design of neural fields \cite{liftingWN2025Chang}.

Such applications, especially interactive ones \cite{cherchi2022interactive}, require computing generalized winding numbers accurately and quickly. Unfortunately, most methods are either accurate but slow \cite{Jacobson13winding,Martens2025WindingNumberOneShot} or fast but approximate \cite{Barill2018FW}. While some of the recent methods do satisfy both requirements, they only target particular parametric surfaces \cite{Spainhour2026} or work solely in 2D \cite{Spainhour2024WN, Liu2025ComplexWN}.

We tackle this challenge by introducing a new method to compute generalized winding numbers on arbitrary meshes and parametric surfaces in 3D, that is both fast and accurate. For meshes, on a CPU, our method is on average $22\times$ (up to three orders of magnitude) faster than the state-of-the-art exact hierarchical method \cite{Jacobson13winding} and is on average $3\times$ faster than the approximation method \cite{Barill2018FW}, while being precise. Our method is embarrassingly parallel and easy to implement on both CPU and GPU. On a GPU, our method is on average $13\times$ faster than a parallelization of the na\"ive method in \cite{Jacobson13winding}. For parametric surfaces, our method is $5.6\times$ faster than the state-of-the-art method by \cite{Spainhour2026}, with the same accuracy.
\begin{figure*}[t]
	\centering
	\includegraphics{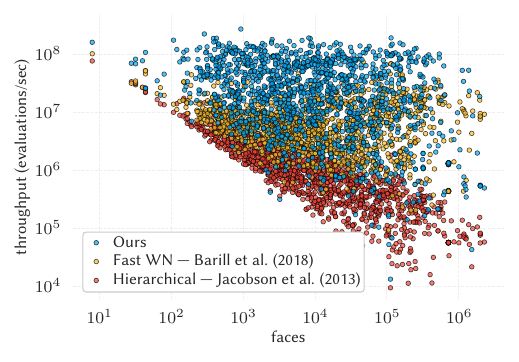}
	\hfill
	\includegraphics{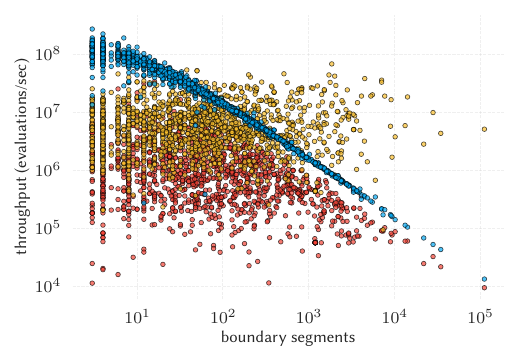}
	\caption{
		Performance comparison (CPU) on meshes: Log-log plot of throughput (in query evaluations per second, the higher the better) vs. number of boundary segments (left) or number of faces (right).
		The fractional part (our bottleneck) of our method scales linearly with the number of boundary segments.
		We compare with \cite{Jacobson13winding} and \cite{Barill2018FW} (order 2). 
		We report evaluation time without precomputation for all methods. }
	\label{fig:performance_boundary}
	\label{fig:performance_faces}
\end{figure*}
Our core insight leading to this efficient algorithm is that a generalized winding number can be expressed as a sum of two terms: the number of signed ray-surface intersections and a surface boundary integral. For both meshes and parametric surfaces, both terms can be computed efficiently and precisely. As we show, the boundary integral is curiously related to the geometry of the surface boundary projected onto a unit sphere: intuitively, it's its total geodesic curvature plus its turning number. This elegant expression gives geometric insight, but more importantly can be directly discretized for meshes and parametric surfaces, and allows us to avoid computing surface integrals \cite{Jacobson13winding} or spherical arrangements \cite{Martens2025WindingNumberOneShot}, yielding a simple-to-implement and fully parallelizable method. Our method is equally efficient for surfaces of arbitrary genus, containing multiple boundaries, or non-manifold geometry.

We extensively validate the performance, robustness, and accuracy of our method, and compare it to the state-of-the-art methods on the Thingi10K dataset \cite{Thingi10K} for meshes and on the dataset from \cite{Spainhour2026, llnl_axom_data} for parametric surfaces.

Our contributions are:

\begin{itemize}
	\item a theoretical result showing that the generalized winding number can be expressed via a signed number of intersections of a ray with the target surface and a line integral of a surface boundary projected onto a unit sphere and
	\item a novel method computing 3D generalized winding numbers for both meshes and parametric surfaces that is significantly faster than previous methods while still being precise.
\end{itemize}

\section{Related Work}

\subsection{Generalized Winding Numbers}

The winding number \cite{meisterGeneraliaGenesiFigurarum1769} is an integer capturing how many times a curve in $\Rtwo$ winds around a query point. The winding number can also be seen as the degree of a covering map, a particular harmonic function, or a particular complex integral \cite{needhamVisualComplexAnalysis1997}, among other interpretations \cite{PerspectivesWindingNumbers2023}.

\citet{Jacobson13winding} introduce generalized winding numbers (GWNs) that extend winding numbers to curves and surfaces with boundary, defined as a signed solid angle, and discretize it to polylines and meshes. Their 3D algorithm computes GWN as the sum of the signed spherical triangle areas, either na\"ively processing each triangle or via an efficient hierarchical algorithm. Our method is significantly faster than both their methods, with similar precision (Sec. \ref{sec:results_performance}). 

\citet{Barill2018FW} introduce a method to compute GWNs for point clouds and triangle soups that can also be used as a fast approximation of GWNs for meshes. As an approximation, their method is incompatible with applications requiring precision. We show that our method is $3\times$ faster than theirs on a CPU, while being accurate to an arbitrary precision (Sec.~\ref{sec:results_performance}).

The One-Shot method \cite{Martens2025WindingNumberOneShot} shows that GWN can be computed via a ray-surface intersection and a sum of signed areas of spherical regions. While their method is efficient for simple boundaries, it requires a computation of a spherical arrangement, which is expensive even for a single boundary component and can have combinatorial complexity for multiple boundary components. Our method also leverages ray-surface intersections, but the rest of our computation is a simple line integral, so 
our method is orders of magnitude faster than theirs (Sec.~\ref{sec:results_performance}).


A few recent papers focus on computing GWNs for parametric curves and surfaces \cite{Spainhour2024WN, Liu2025ComplexWN,Spainhour2026}. \citet{Spainhour2024WN}, \citet{Liu2025ComplexWN}, and \citet{bao2025fastrobustpointcontainment} deal with 2D curves only; our focus is 3D surfaces. In a recent work, \citet{Spainhour2026} formulate GWN via a boundary integral and use ray-surface intersections for query points near the surface. In this form they are similar to our method, yet our expressions are different. Their method only works for parametric surfaces whereas ours also works on meshes. We compare our performances for parametric surfaces in Sec.~\ref{sec:results_performance}; we use their ray-parametric surface intersection algorithm and caching.

A closely related concept to a winding number is a signed solid angle of a space curve, which is the winding number up to a multiple of $4\pi$ (without the normalization factor) \cite{binyshMaxwellTheorySolid2018,chern2024area}. In particular, despite the $4 \pi$ modulo difference, our Lemma \ref{lemma} is closely related to Eq.~8 in \citet{binyshMaxwellTheorySolid2018}, but their formula uses the number of curve self-intersections and so does not lead to an efficient discretization. The main result in \cite{Rogen} is also related to our Lemma \ref{lemma}. In contrast, our formulation is an explicit boundary integral, leading to an efficient discretization and the algorithm.

\subsection{Containment Queries and Ray Casting}
Our method is related to the works that use ray casting for containment queries, for example, to compute voxelizations of geometry with various degeneracies \cite{nooruddin2003simplification,Houston2003}. For instance \citet{nooruddin2003simplification} leverage several ray-surface intersections for a robust notion of insidedness. Our method only needs a single ray-surface intersection per point or fewer for query points in regular grids, such as in voxelization scenarios. For trimmed parametric surfaces, the ray-surface intersection requires a containment test with respect to the trimming curves. It is often done via parity tests \cite{nishita1990ray}, signed area \cite{Martin01012000}, or winding number \cite{Spainhour2024WN}.
\section{Background on Winding Numbers}

For brevity, we omit the word ``generalized'' in description of generalized winding numbers.
We are interested in two different perspectives on winding numbers: as an integral of signed solid angle and as an integral of a signed number of intersections.
 
\subsection{Signed Solid Angle}
Given an oriented 2-manifold (with or without boundary) $\M$ embedded in $\Rthree$ and a point $\vec{p} \in \Rthree \; \backslash \; \M$, the GWN is defined via the signed solid angle (Fig. \ref{fig:solid_angle_vs_chi}):

\begin{equation}
	\label{eq:wn_solid_angle}
	w_\M(\vec{p}) = \frac{1}{4 \pi} \int_\M d\Omega(\vec{p}).
\end{equation}

\citet{Jacobson13winding} compute this integral for a triangle mesh as the sum of signed areas of spherical triangles, i.e., mesh triangles projected\footnote{Here and everywhere else we mean radial projection $\pi_{\vec{p}}(x) = \vec{p}+\frac{\vec{x}-\vec{p}}{\norm{\vec{x}-\vec{p}}}$.} onto a unit sphere $S^2_{\vec{p}}$:

\begin{equation}
	\label{eq:solid_angle_sum}
	w_\M(\vec{p}) = \frac{1}{4 \pi} \sum_{\T \in \M} \textrm{Area}_{S^2_{\vec{p}}}(\T).
\end{equation}


\subsection{Ray Casting}
\label{sec:raycasting}

Alternatively, GWN is the average number of signed intersections $\chi$ between $\M$ and a ray leaving the query point, where an intersection gives $+1$ when the ray has a positive dot product with the surface normal, and $-1$ if negative\footnote{Under simple assumptions, tangential rays with zero dot product form a measure-zero set and thus do not influence the integral.} \cite{PerspectivesWindingNumbers2023}. It can be expressed as an integral over $ S^2_{\vec{p}}$, the unit sphere centered at $\vec{p}$ (Fig. \ref{fig:solid_angle_vs_chi}b):  

\begin{equation}
	\label{eq:wn_int_over_chi}
	w_\M(\vec{p}) = \frac{1}{4 \pi} \int_{\vec{q} \in S^2_{\vec{p}}} \chi(\vec{q}) \; dA.
\end{equation}

\begin{figure}
	\centering
	\includegraphics[width=\linewidth]{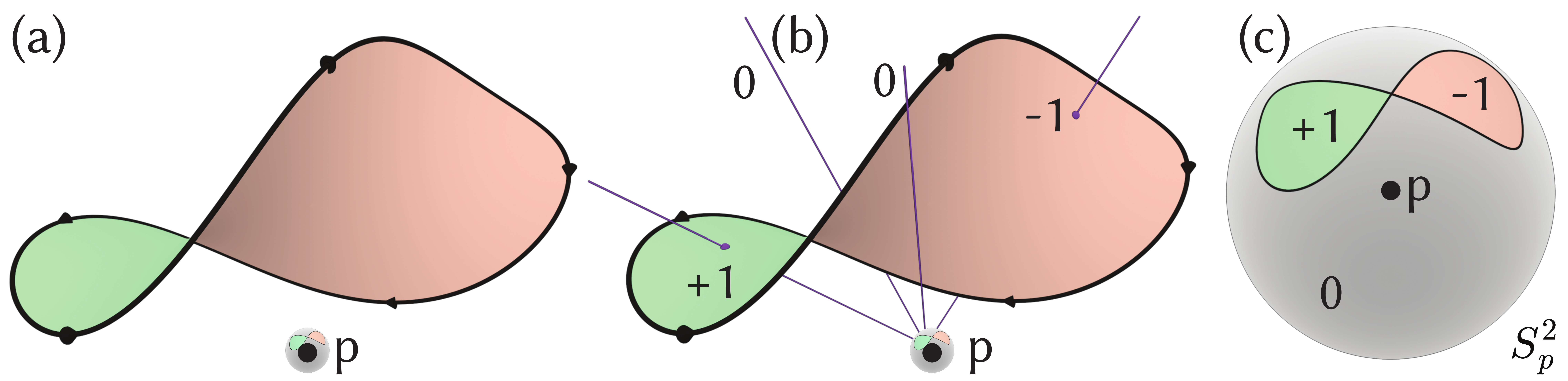}
	\caption{For a given surface, a generalized winding number can be interpreted either as an integral of the signed solid angle (a) or as an average number of signed intersections of rays evenly covering a unit sphere (b). The number of signed intersections, when visualized over the unit sphere, changes by $\pm1$ only when crossing the projected surface boundary, otherwise it is a constant (c).}
	\label{fig:solid_angle_vs_chi}
\end{figure}

As observed in \citet{Martens2025WindingNumberOneShot}, $\chi(\vec{q})$ is constant in each region $\Omega_i$ formed by projecting the curve onto the sphere (Fig.~\ref{fig:solid_angle_vs_chi}c), so using $i^{\text{th}}$ region's $\chi_i$ and its area $A_i$, Eq.~\ref{eq:wn_int_over_chi} is equal to: 

\begin{equation}
	\label{eq:wn_chi_sum}
	w_\M(p) = \frac{1}{4 \pi}\sum_{i=1}^n \chi_i A_i.
\end{equation}
As they discuss, $\chi$ changes by $\pm 1$ each time we cross the projected curve, so knowing one arbitrary $\chi_i$ and the spherical arrangement of regions is enough to compute all other $\chi_j$.

\subsection{Spherical Geometry}
Spherical polygons, and in particular, triangles play a central role in our formulation; we only consider the unit sphere $S^2$. A spherical polygon is a polygon with vertices on a sphere connected with geodesics, which are arcs of great circles. To measure the area of a $N$-sided spherical polygon, Girard's Theorem asserts that one can compute the sum of interior angles minus $(N-2) \pi$. For instance, the (unsigned) area of a spherical triangle $\mathcal{T}$ with the angles $\alpha, \beta, \gamma>0$ is equal to the \emph{spherical excess}:

\begin{equation}
	\label{eq:excess}
	\textrm{UnsignedArea}_{S^2}\left(\mathcal{T}\right) = \alpha + \beta + \gamma - \pi \geq 0.
\end{equation}

 In the general case, three points on the sphere can denote one of the two triangles, a small one or a big one, depending on the orientation. The orientation therefore defines whether this formula computes either the area of the smaller one or the area of its complement, adding up to $4\pi$ in total.
 
 From the computational perspective, a slightly more efficient formula for the \emph{signed} area of a spherical triangle $\mathcal{T}$ with the vertices $\vec{v_0}, \vec{v_1}, \vec{v_2}$ is provided in \citet{Oosterom1983SphericalTriangleArea}:
 \begin{equation}
 	\textrm{Area}_{S^2}\left(\mathcal{T}\right) = 2 \; \mathrm{atan2}\left( \vec{v_0} \cdot (\vec{v_1} \times \vec{v_2}), 1 + \vec{v_0} \cdot \vec{v_1} + \vec{v_0} \cdot \vec{v_2} + \vec{v_1} \cdot \vec{v_2} \right). \end{equation}
 
An (unsigned) area of an arbitrary simply connected spherical region $\Omega$ with the boundary $\partial \Omega$ with $N$ corners connected by curves that are not necessarily geodesic can be computed via the Gauss-Bonnet Theorem \cite{do_carmo_differential_2016}:

\begin{equation}
	\textrm{UnsignedArea}_{S^2}(\Omega) = 2\pi - \int_{\partial \Omega} k_g \, ds - \sum_{i=1}^{N} \theta_i,
	\label{eq:gb_area}
\end{equation}
where $k_g$ is the geodesic curvature of $\partial \Omega$ and $\theta_i$ are the exterior turning angles at the corners. The spherical triangle area formula (Eq.~\ref{eq:excess}) follows from $k_g=0$, as all sides are geodesics; note also that (Eq.~\ref{eq:excess}) uses interior triangle angles, while (Eq.~\ref{eq:gb_area}) uses exterior turning angles.
 
\section{Method}
Although our method is simple to understand and implement for meshes, it has a solid theoretical foundation in the continuous setting that allows us to apply our method to the parametric surfaces. Thus, we first informally describe the intuition for meshes (Sec.~\ref{sec:intuition}), then formally develop the continuous formulation (Sec.~\ref{sec:continuous}), finishing with the discretization and the full algorithm, both for meshes and parametric surfaces (Sec.~\ref{sec:discrete}).

\subsection{Intuition for Meshes}
\label{sec:solid_angle_discretization}
\label{sec:intuition}
\begin{figure}
	\centering
	\includegraphics[width=\linewidth]{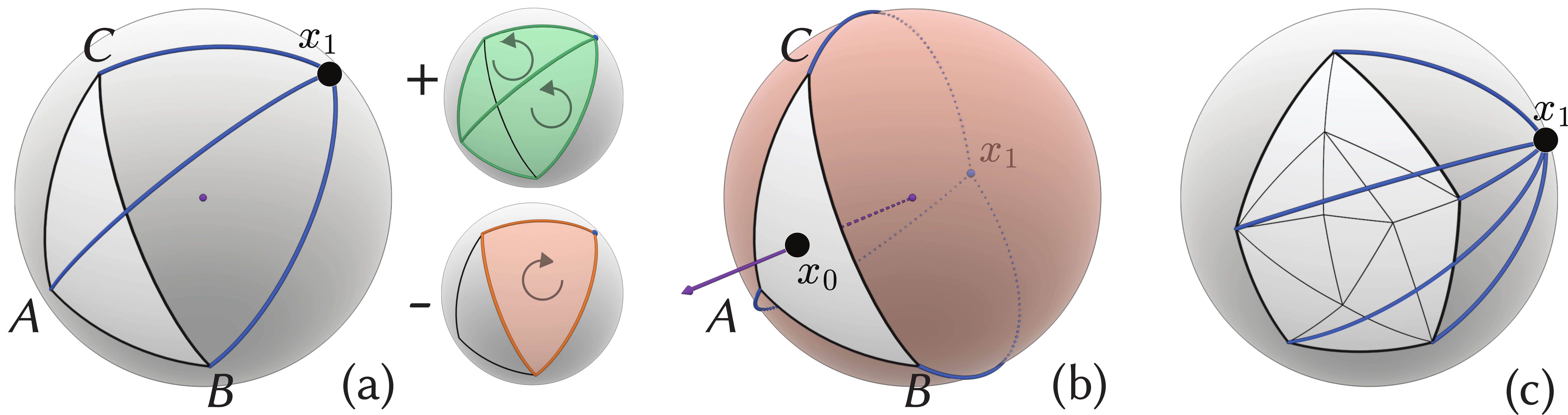}
	\caption{For a spherical triangle $\Delta ABC$ we form three auxiliary triangles $\Delta ABx_1, \Delta BCx_1, \Delta CAx_1$, connecting each edge with an arbitrary point $\North$ with the shortest geodesics (blue). Depending on the location of $\North$, the sum of signed areas of these auxiliary triangles is either (a) equal to the sum of the triangle $\Delta ABC$ (inset, two green positive areas, one negative red area) or (b) to the area of its complementary region on the sphere --- it depends on whether the antipodal point $\South = -\North$ is inside or outside the triangle. Summing this for a whole mesh, most areas cancel out, leaving the sum of areas of the auxiliary triangles adjacent to the boundary segments (c).}
	\label{fig:intuition}
\end{figure}

Inspired by the shoelace formula \cite{braden1986surveyor} to compute an area of a simple planar polygon as a sum of signed areas of trapezoids, we develop a similar formula for spherical polygons. We observe that the signed area of each spherical triangle $\SphericalArea(\T)$ (Eq.~\ref{eq:solid_angle_sum}) can be expressed via a sum of signed areas of three \emph{auxiliary} triangles, formed by connecting each triangle vertex with an arbitrary point $\North \in \StwoP$ via the shortest geodesics (Fig.~\ref{fig:intuition}). 

Depending on the location of $\North$ with respect to the triangle, the sum of three auxiliary triangle areas is equal to the area of either the triangle itself or of its complement (Fig.~\ref{fig:intuition}ab). This choice depends on whether the \textit{antipodal} point $\South \coloneq -\North$ is inside or outside the triangle. For a spherical triangle $\Delta ABC$, this sum is $\SphericalArea(\T)$ if $\South \in \Delta ABC$ and $\pm 4\pi + \SphericalArea(\T)$ otherwise, $``+"$ for a negatively oriented triangle and $``-"$ for a positively oriented one.
 
To compute the winding number, we can sum up those three auxiliary triangle areas for each mesh triangle, yielding the sought projected area plus $4\pi \chi(\South), \chi(\South) \in \mathbb{Z}$. As we add a $\pm 4\pi$ each time $\South$ is in a triangle, $\chi(\South)$ is equal to the number of times the surface projection covers (taking orientation into account) $x_0$. Thus, to compute $\chi(\South)$ one can shoot a ray $\vec{p}\South$ and count the number of signed intersections with the surface (Fig.~\ref{fig:solid_angle_vs_chi}).

Thus,
\begin{equation}
	\sum_i \SphericalArea(\T_i) = 4\pi \chi(\South) + \sum_i \sum_{j=1}^3 \SphericalArea(\T^j_i).
	\label{eq:sum_of_areas}
\end{equation}

Finally, since all the interior triangle edges have two equivalent adjacent auxiliary triangles with opposite signs, their areas cancel out. Combining Eqs.~\ref{eq:solid_angle_sum} and \ref{eq:sum_of_areas}, the GWN can be computed as:

\begin{equation}
	\label{eq:anchored_triangle_sum}
	w_M(\vec{p}) = \chi(\South) + \frac{1}{4 \pi}\sum_{(\vec{a},\vec{b}) \in \MeshBoundary} \SphericalArea(\vec{a}, \vec{b}, \North),
\end{equation}
where $\chi(\South)$ is the number of signed intersections of a ray $\vec{p}\South$ with the mesh (Fig.~\ref{fig:intuition}c).  For a closed mesh it is just an inside-outside test. 

This formula is the entirety of our method for meshes (Alg.~\ref{alg:wnr}).

\begin{algorithm}
	\caption{GWN at a point $\vec{p}$ for a mesh $\Mesh$ with boundary $\partial \Mesh$ }
	\label{alg:wnr}
	
	\KwIn{$p$, $\Mesh$}
	$\vec{x_0} \gets \mathrm{RandomUnitVector()}$ \tcp*{Arbitrary ray direction}
	$\vec{x_1} \gets -\vec{x_0}$ \tcp*{Antipodal point}
	
	$Area \gets 0$\;
	\ForEach{$(\vec{a},\vec{b}) \in \partial \Mesh \;$}{
		$\vec{A} \gets \frac{\vec{a} - \vec{p}}{\lVert \vec{a} - \vec{p}\rVert}$ \tcp*{Unit Sphere Projection}
		$\vec{B} \gets \frac{\vec{b} - \vec{p}}{\lVert \vec{b} - \vec{p}\rVert}$\;
		$Area \gets Area + \mathrm{SignedSphericalTriArea}(\vec{A}, \vec{B}, \vec{x_1})$\;
	}
	$\chi \gets \mathrm{SignedIntersectionNumber}(\vec{p},\vec{x_0},\Mesh)$ \;
	\Return{$\chi + \dfrac{Area}{4\pi}$} \tcp*{Eq.~\ref{eq:anchored_triangle_sum}}
	
	\Fn{SignedSphericalTriArea($\vec{v_0},\vec{v_1},\vec{v_2}$)}{
		\Return{$2 \: \mathrm{atan2}\left( \vec{v_0} \cdot  (\vec{v_1} \times \vec{v_2}), 1 + \vec{v_0} \cdot \vec{v_1} + \vec{v_0} \cdot \vec{v_2} + \vec{v_1}\cdot \vec{v_2} \right)$}
	}
\end{algorithm}

\subsection{Continuous Formulation}
\label{sec:continuous}
Our goal is to compute the sum in Eq.~\ref{eq:wn_chi_sum}. We first note $\chi(\vec{q})$ follows the Alexander numbering \cite{alexander1928topological}: its value is constant in every region, increments by $+1$ or $-1$ whenever we cross the projected boundary $\Curve = \mathrm{Proj}_{S_p^2}\partial \M$ from the right or the left, respectively (Fig.~\ref{fig:solid_angle_vs_chi}c) \cite{PerspectivesWindingNumbers2023,Martens2025WindingNumberOneShot}. All integer-valued functions following the Alexander numbering on a sphere for a given closed curve $\Curve$ are the same up to a constant shift: $f(\vec{q}) = g(\vec{q}) + c,$ where $c\in \mathbb{Z}$. We call such functions \emph{spherical winding numbers} \cite{PerspectivesWindingNumbers2023} (Fig.~\ref{fig:spherical_wn}).

\begin{figure}
	\includegraphics[width=0.45\columnwidth]{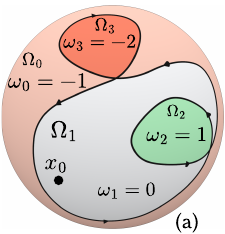}
	\hfill
	\includegraphics[width=0.45\columnwidth]{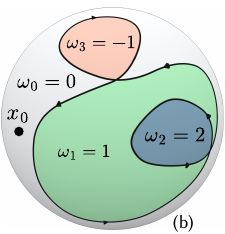}
	\caption{Given a curve $\Gamma$ dividing the sphere into regions $\Omega_i$, we call integer-valued functions $\Weight$ satisfying the Alexander numbering spherical winding numbers (a,b). They can only differ by an integer and thus can be uniquely defined by a point $x_0$ where $\Weight=0$.}
	\label{fig:spherical_wn}
\end{figure}

For an arbitrary point $\South$ on the sphere $\StwoP$, let us take the unique spherical winding number $\Weight_{\South}(\vec{q})$ such that $\Weight_{\South}(\South)=0$ (Fig.~\ref{fig:spherical_wn}). As a shorthand notation, its value in the $i^{\text{th}}$ region is $\Weight_i$; we will also omit the subscript $\South$ where there is no ambiguity. Then $\chi_i = \chi(\South)+\Weight_i$ and the sum in Eq.~\ref{eq:wn_chi_sum} can be rewritten as:
\begin{equation}	
	\label{eq:int_with_offset}
w_\M(\vec{p}) = \frac{1}{4 \pi}\sum_{i=1}^n \chi_i A_i = \chi(\South) + \frac{1}{4 \pi}\sum_{i=1}^n \Weight_i A_i,
\end{equation}
as $\sum_i A_i = 4\pi$, since the regions cover the sphere once.

The first term can be computed via ray intersections (Sec.~\ref{sec:raycasting}). Below we prove a result that allows us to efficiently compute the second term via line integrals. A related statement has been mentioned by \cite{generalizedGaussBonnet}, albeit without the full proof. 

First, we generalize the standard notion of a turning number on a plane, also known as a Whitney index, to a surface, following \cite{burmanWhitneysFormulasCurves2011}. For a vector field\footnote{Vector fields are assumed to be tangential, i.e., sections of the tangent bundle.} $\VF$ on $\StwoP$ with no zeros on a smooth\footnote{We later relax the smoothness requirement in Prop.~\ref{prop}.} curve $\Curve \subset S_p^2$ that might have transversal self-intersections,  the curve's Whitney index $\text{ind}(\Curve,\VF)$ is the number of rotations of the tangent $\Curve'(t)$ relative to $\VF$. More formally, if $\theta$ measures the angle between $\VF$ and $\Curve'$, then
\begin{equation}
	\ind(\Curve,\VF) = \frac{1}{2\pi} \int_\Curve d\theta.
	\label{eq:index}
\end{equation}

\begin{figure}
	\includegraphics[width=\linewidth]{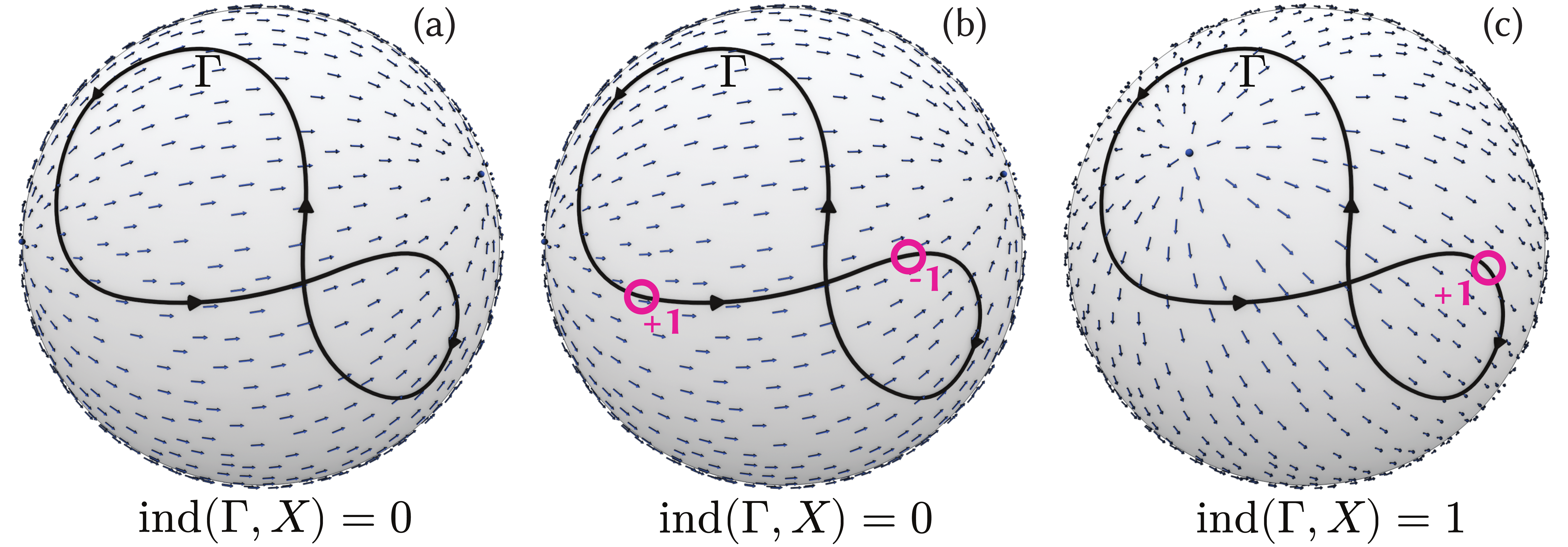}
	\caption{For a given vector field $\VF$ on a sphere, the Whitney index of a curve $\Curve$ can be defined as the number of rotations of the curve tangent $\Curve'$ relative to $\VF$ (a). Equivalently, it can be computed as the number of the points where the tangent is codirectional with $\VF$, counted as $+1$ or $-1$ depending on whether $\Curve'$ turns counterclockwise or clockwise around those points, respectively (b). The index depends on the vector field choice (c).}
	\label{fig:curve_index}
\end{figure}

The index can be equivalently calculated as the number of points at which $\Curve'(t)$ looks in the direction of $\VF$, i.e., they are codirectional; each such point is counted with a positive sign if $\Curve'(t)$ turns counter-clockwise relative to $\VF$ in a neighborhood of $\Curve(t)$ and with a negative sign otherwise (Fig.~\ref{fig:curve_index}ab). Note that the index depends on the choice of vector field $\VF$ \cite{mcintyreNewFormulaWinding1993a} (Fig.~\ref{fig:curve_index}c). 

Now we formulate a result which will form the core of our algorithm, first for a rather constrained vector field $\VF$, then by relaxing this constraint. Let $\VF$ be a vector field on the unit sphere $\StwoP$ with singularities of index 1 only (by Poincar\'e-Hopf, there will be exactly two of them), located in the regions with $\Weight_i=0$.

\begin{lemma}
	\label{lemma}
	Under the previous assumptions,
	\begin{equation} \sum_{i=1}^n \Weight_i A_i = 2\pi \ind(\Curve,\VF) - \int_\Curve k_g \; ds,
	\end{equation}
	where $k_g$ is the geodesic curvature of curve $\Curve$.
\end{lemma}

The proof is in the Supplementary. To support piecewise smooth curves, typical for boundaries of both meshes and parametric surfaces, we generalize index (Eq.~\ref{eq:index}): If the boundary $\Curve = \cup_i \Curve_i$, where each $\Curve_i$ is smooth, then
\begin{equation}
 	\text{ind}(\Curve,\VF) = \frac{1}{2\pi} \left( \int_\Curve d\theta  + \sum_i \theta_i \right),
 	\label{eq:index_piecewise}
\end{equation}
where $\theta_i$ is the turning angle between the $i^\mathrm{th}$ and the next curve. The geodesic curvature integral should be adjusted accordingly. Now we are ready to formulate our main result that allows for fast computations of winding numbers, where none of the terms require surface integrals:

\begin{proposition} For a vector field $\VF$ on a unit sphere $\StwoP$ with two singularities $\South, \North$ of index 1,
	\begin{equation}w_\M(\vec{p}) = \chi(\South) + \frac{1}{2} (\ind(\Curve,\VF)+\Weight(\North)) - \frac{1}{4\pi}\left( \int_\Curve k_g \; ds + \sum_i \theta_i \right).
	\label{eq:prop}
	\end{equation}
	\label{prop}
\end{proposition}
\begin{proof}
	If $\Weight(\North)=0$, this follows from Lemma \ref{lemma} and Eq.~\ref{eq:int_with_offset}. To prove the general case, we observe that as we move the singularity point $\North$ across the curve, $\ind(\Curve,\VF)$ decreases by one  when we enter the curve from the left and increases by one if we enter the curve from the right. The Alexander numbering changes in precisely the opposite way, i.e., decreasing and increasing by one respectively, so $\ind(\Curve,\VF) + \Weight(\North)$ does not depend on the singularities' locations.
	\end{proof}
	
Finally,  we observe that if we substitute the definition of the index (Eq.~\ref{eq:index_piecewise}) into Eq.~\ref{eq:prop}, the terms $\sum_i \theta_i$ cancel. Furthermore, the quantity $d\theta$ in Eq.~\ref{eq:index_piecewise} can be decomposed into two parts: the rotation of the curve tangent, and the rotation of the vector field with respect to that tangent. Along the curve, the tangent rotates at the rate $k_g$ with respect to the parallel transport. If the vector field makes the angle $\eta$ with respect to the parallel transport along the curve, then $d\theta = k_g \; ds - d\eta$ (c.f. Ch.4.4, Lemma 2 in \cite{do_carmo_differential_2016}). As a result, $\frac{1}{2} \ind(\Curve,\VF) - \frac{1}{4\pi}\left( \int_\Curve k_g \; ds + \sum_i \theta_i \right) = -\frac{1}{4\pi}\int_\Gamma d\eta$. This yields the final formula:

\begin{proposition}
	\begin{equation}
	w_\M(\vec{p}) = \chi(\South) +\frac{1}{2}\Weight(\North) - \frac{1}{4\pi} \int_\Gamma d\eta.
	\label{eq:simpler_prop}
\end{equation}
\label{prop:simpler}
\end{proposition}
 
Simply put, the GWN consists of a signed intersection number at one point and two boundary terms: a spherical winding number at another point and the total rotation, relative to parallel transport, of a vector field with singularities at these two points. We use Prop.~\ref{prop} and Prop.~\ref{prop:simpler} directly for efficient algorithms.
	

\subsection{Triangle Meshes}
\label{sec:discrete}
For a triangle mesh, we follow our discrete intuition (Sec.~\ref{sec:intuition}) and show that Alg.~\ref{alg:wnr} computes Eq.~\ref{eq:prop} (and therefore Eq.~\ref{eq:simpler_prop}).

We pick a vector field that is the gradient of a height function, with the sink and the source at $\South$ and its antipodal point $\North$ (Fig.~\ref{fig:discretization}a), respectively. Then applying the formula in Eq.~\ref{eq:prop} to the mesh yields Eq.~\ref{eq:anchored_triangle_sum}. To see that, note that $\Curve$ for a mesh is a spherical polygon. For our choice of the vector field, its flowlines form the sides of the geodesic triangles connecting each segment of the polyline with $\North$ as the vertex, with angles $\alpha_i, \beta_i, \gamma_i > 0$ and the triangle orientation $\delta_i = \pm1$ (Fig.~\ref{fig:discretization}b), so:

\begin{figure}
	\includegraphics[width=0.35\linewidth]{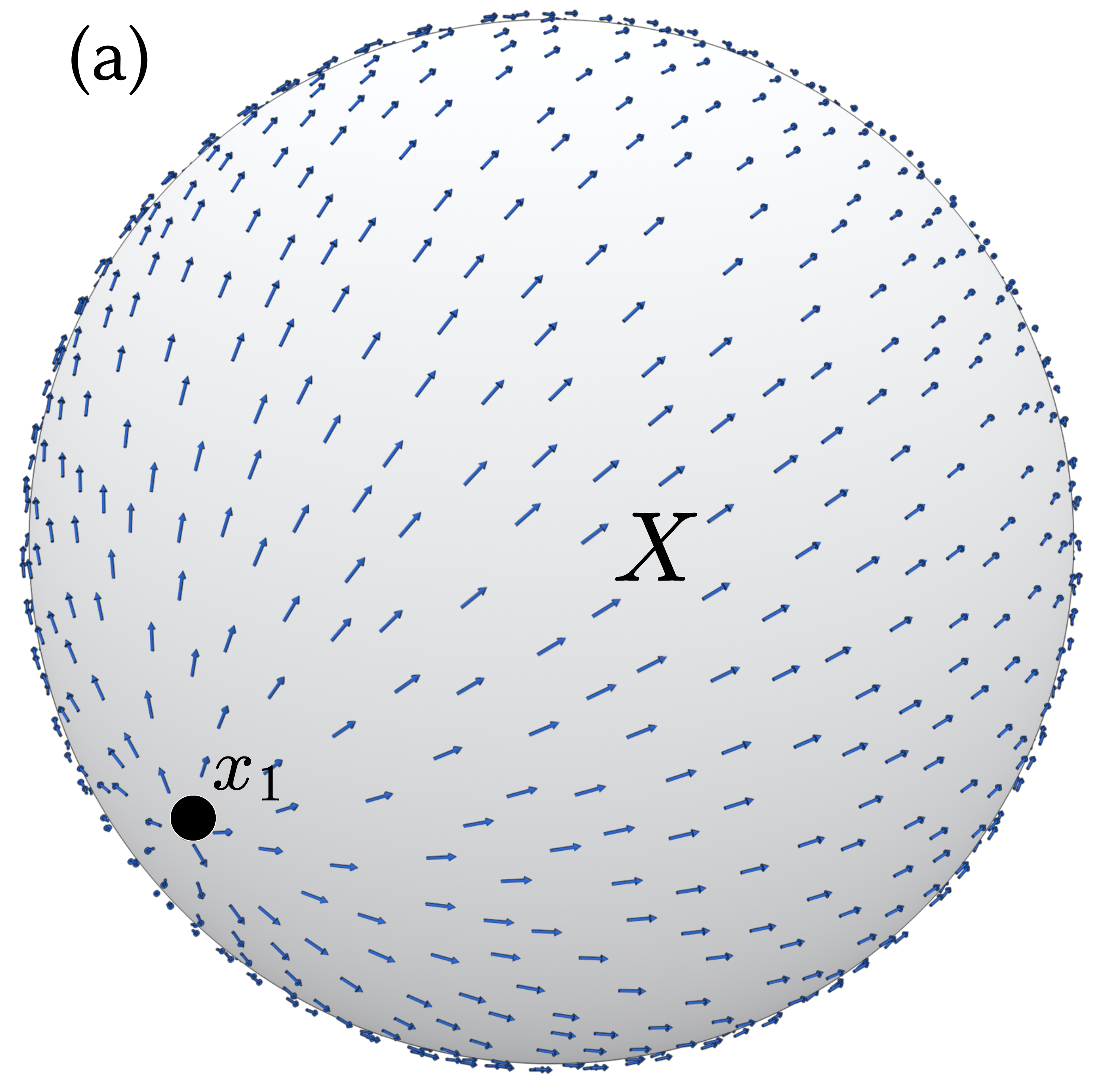}
	\hspace{2.5em}
	\includegraphics[width=0.35\linewidth]{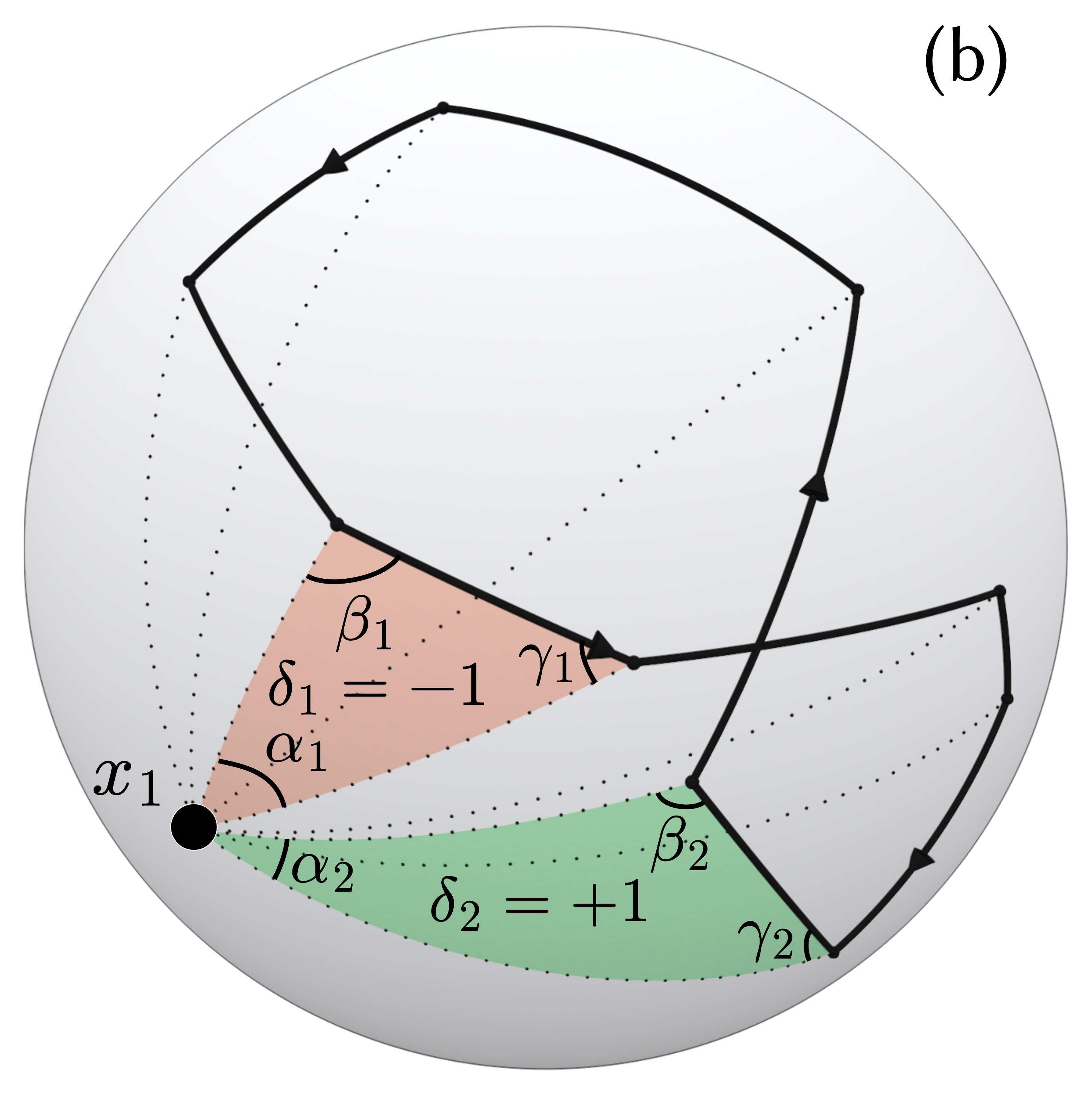}
	\caption{To discretize our formula for meshes, we pick the vector field $X$ to be the ``gradient of height'', i.e., a vector field parallel to the ``meridians'' with the ``poles'', a source and a sink, at antipodal points $\South$ and $\North = -\South$ ((a), $\South$ is on the invisible side). Connecting each segment of the boundary with $\North$ via accordingly oriented flowlines of this field (dashed blue lines), which are shortest geodesics, we get geodesic triangles, some positively oriented, some negatively (b). When summed, the triangle angles $\beta$ and $\gamma$ together measure the geodesic curvature of the curve and its index with respect to this vector field, and $\alpha$ measures the spherical winding number $\Weight_{\South}(\North)$.}
	\label{fig:discretization}
\end{figure}
\begin{equation}
	\label{eq:betas_and_gammas}
	2\pi \ind(\gamma,X) - \int_\gamma k_g ds - \sum_i \theta_i = \sum_i \delta_i (\beta_i + \gamma_i - \pi),
\end{equation}
as $\pi - \beta_i - \gamma_i$ computes the integrated geodesic curvature around the vertex $i$, and the signs change at the points where geodesics are tangent to the curve --- at exactly $2\ind(\Curve,\VF)$ points: $\ind(\Curve,\VF)$ of points of tangency with tangents aligned with the vector field and the same number of points with the opposite tangent orientations. Each sign change adds a $\pi$, yielding the formula. Finally,

\begin{equation}
	\label{eq:sum_of_alphas}
	\Weight_{\South}(\North) = \Weight_{-\North}(\North) = \sum_i \delta_i \alpha_i,
\end{equation}
which can be seen from a stereographic projection $\pi_{\South} : S^2 \; \backslash \; \{\South\} \rightarrow \Rtwo$ with the center of projection at $\South$: All geodesics emanating from $\North$ become straight line segments, as no shortest geodesic from $\North$ passes through its antipodal point $\South$. The map $\pi_{\South}$ is conformal, so the sum $\sum_i \alpha_i$ is the planar winding number around $\pi_{\South}(\North)$. Finally, $\Weight(\North)$ can be computed as the number of signed intersections of a geodesic connecting some point of the spherical region containing $\South$ (where $\Weight = 0$) with $\North$. This geodesic maps to a straight line segment connecting the exterior of the shape (where the planar winding number is 0) to $\pi_{\South}(\North)$. As the stereographic projection is an orientation-preserving homeomorphism and the number of signed intersections is invariant under such maps, the pullback of the planar winding number is equal to $\Weight$. As the signed area of each spherical triangle is $\delta_i (\alpha_i + \beta_i + \gamma_i - \pi)$, the sum of all areas computed by Alg.~\ref{alg:wnr} computes the formula in Prop.~\ref{prop}.


\subsubsection{Ray-Mesh Intersections}
\label{sec:mesh_intersect}
For multiple query points, all the query points on the same line can reuse the intersection ray.  So for a voxelization scenario $W \times H \times D$ ($W\leq H \leq D$),  we only need $W \times H$ rays parallel to the $z$ axis. Assuming we shoot in the positive $z$ direction, once all intersections are computed, the signed intersection number for each query point can be computed by following the ray and adding $\pm 1$ for each intersection, depending on whether the intersected triangle is front-facing or back-facing, respectively.

\begin{figure}
	\includegraphics[width=\columnwidth]{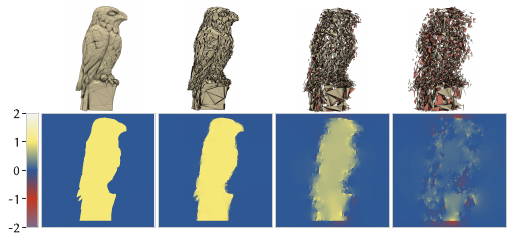}
	\captionof{figure}{Our method works equally well on closed meshes (left), where it becomes a simple inside-outside test, and triangle soups. Maltese Falcon by colinfizgig (Thing ID: 46631). }
	\label{fig:results_soups}
\end{figure}

\subsection{Parametric Surfaces}
\label{sec:parametric}
For a parametric surface, such as NURBS, we directly evaluate Eq.~\ref{eq:simpler_prop}. We compute the $\chi(\South)$ via ray-surface intersection in Sec.~\ref{sec:parametric_int}.

The second term in Eq.~\ref{eq:simpler_prop}, the spherical winding number $\omega_{\South}(\North)$, can be computed as the total number of rotations of $\vec{v}(t)$, the tangent of the shortest geodesic from $\North$ towards $\Gamma(t)$. This can be computed as the radius vector $\vec{\Gamma}(t)-\North$ projected onto the tangent plane at $\North$; note that this is an equivalent of Eq.~\ref{eq:sum_of_alphas}, the same explanation applies. We compute the rate of rotation of $\vec{v}(t)$, denoted as $w(t)$, yielding the spherical winding number.

To compute the last term, we evaluate the vector field in Fig.~\ref{fig:discretization} at each point $\vec{\Gamma}(t)$ as the orthogonal projection of the vector $\South-\vec{\Gamma}(t)$ onto the tangent plane with the normal $\vec{\Gamma}(t)$, so the vector field is $\vec{X}_{\vec{\Gamma}} = \South-\vec{\Gamma}(t) - \left((\South-\vec{\Gamma}(t))\cdot \vec{\Gamma}(t)\right) \vec{\Gamma}(t) = \South - (\South \cdot \vec{\Gamma}(t))\vec{\Gamma}(t)$. To measure $d\eta$, we compute the rate of rotation of the unit vector field $ \frac{\vec{X}_{\vec{\Gamma}(t)}}{\lVert\vec{X}_{\vec{\Gamma}(t)}\rVert}$ in the tangent plane with the normal $\vec{\Gamma}(t)$.

Both of the terms include computing the angular rates of a vector $U(t)$ in a tangent plane with a normal $n$, which we compute as $n \cdot \frac{\vec{U}(t) \times \vec{U}'(t)}{\lVert \vec{U} \rVert^2}$. Having defined those terms, we numerically integrate them via 1D adaptive quadrature (Gauss-Kronrod $GK15$ in our implementation) to evaluate Eq.~\ref{eq:simpler_prop}.

\begin{algorithm}[t]
		\caption{GWN at $\vec{p}$ for a parametric surface $\M$ with boundaries $\vec{\varphi}_i$, $i=1, \ldots, n$.}
		\label{alg:parametric-wnr}
		\DontPrintSemicolon
		\KwIn{$\vec{p}$, $\M$, $\varepsilon > 0$}
		
		$\vec{x}_0 \gets \mathrm{RandomUnitVector()}$ \tcp*{Arbitrary ray direction}
		$\vec{x}_1 \gets -\vec{x}_0$ \tcp*{Antipodal point}
		
		\Fn{$f(\vec{\varphi}, t)$}{
			$\vec{\Gamma} \gets \dfrac{\vec{\varphi}(t) - \vec{p}}{\lVert\vec{\varphi}(t) - \vec{p}\rVert}$ \tcp*{Unit sphere projection}
			
			$\vec{X}_{\Gamma} \gets \vec{x}_0 - \left( \vec{x}_0 \cdot \vec{\Gamma}\right) \vec{\Gamma}$ \tcp*{The vector field at $\vec{\Gamma}(t)$}
			$d\eta \leftarrow \dfrac{\vec{\Gamma} \cdot \left( \vec{X}_{\Gamma} \times X'_{\Gamma} \right)}{\lVert \vec{X}_{\Gamma}\rVert^2}$ \tcp*{The vector field's rotation rate}
			
			$\vec{v} \gets \vec{\Gamma} - \left( \vec{x}_1 \cdot \vec{\Gamma}\right) \vec{x}_1$ \tcp*{Tangent of the geodesic}
			
			$w \leftarrow \dfrac{\vec{x}_1 \cdot \left( \vec{v} \times \vec{v'} \right)}{ \lVert\vec{v}\rVert^2}$ \tcp*{The rotation rate of $\vec{v}$}
			
			\Return $\dfrac{1}{4 \pi}\left(w - d\eta\right)$\;
		}
		
		\Return $\chi + \sum_{i=1}^n \textsc{AdaptiveGaussKronrod15}(f(\vec{\varphi}_i,[0,1]), \varepsilon)$
\end{algorithm}

\subsubsection{Ray-Parametric Surface Intersections}
\label{sec:parametric_int}
For parametric surfaces, in a single query point scenario, we follow \citet{Spainhour2026}: we subdivide surfaces by doing B\'ezier extraction of NURBS patches and bisection of individual B\'ezier patches. We terminate this recursion whenever a patch is approximately bilinear, measured as the Hausdorff distance between the patch and its bilinear approximation being less than $\tau = 10^{-6}$. We then use a closed-form ray intersection test for bilinear patches \cite{reshetov2019cool}. We eliminate all intersections outside the trimming region via raycasting in the parameter domain. For voxelization, we reuse rays (Sec.~\ref{sec:mesh_intersect}). 

Whenever our ray intersection hits too close (within machine precision) to the surface boundary or hits a cusp (surface normal length less than $10^{-4}$) or is tangent to the surface (absolute value of dot product less than $10^{-4}$), we simply discard this ray and shoot another one. Except for the extra intersection test, this does not incur any other cost and makes the implementation more robust.
\section{Results}
\label{sec:results}

\begin{table}[t]
	
		\caption{Performance comparison, per query point, with \cite{Spainhour2026}, on their parametric dataset. Our method integrates over the surface boundary curves; theirs integrates over the patch boundaries.}
	\centering
	\setlength{\tabcolsep}{2.5pt}
	\renewcommand{\arraystretch}{1.0}
	\begin{tabular}{lrrrr}
\toprule
Model Name &
\multicolumn{2}{c}{Timings (ms)} &
\multicolumn{2}{c}{\#Boundary Curves} \\
\cmidrule(lr){2-3}\cmidrule(lr){4-5}
&
Ours & Spainhour &
Patches & Surface \\
\midrule
		Bobbin & \textbf{0.247} & 0.521 & 76 & 13 \\ 
		Bolt & \textbf{2.129} & 2.979 & 1161 & 194 \\
		Boxed Sphere & \textbf{0.284} & 0.971 & 968 & 176 \\
		Connector & 4.109 & \textbf{2.860} & 334 & 126 \\
		Gear & \textbf{19.476} & 29.742 & 10256 & 1160 \\
		Joint & \textbf{0.084} & 1.305 & 564 & 52 \\
		Lamp & \textbf{0.008} & 0.606 & 214 & 0 \\
		Nut & \textbf{0.074} & 0.129 & 66 & 14 \\
		Open Cylinder & \textbf{0.041} & 0.068 & 15 & 13 \\
		Pipe & \textbf{3.516} & 7.076 & 563 & 80 \\
		Sliced Cylinder & \textbf{0.094} & 0.114 & 18 & 14 \\
		Slide & \textbf{0.057} & 3.418 & 1776 & 0 \\
		Spring & \textbf{0.011} & 0.206 & 256 & 4 \\
		2-Patch Spring & \textbf{0.015} & 0.709 & 64 & 4 \\
		Teardrop & \textbf{0.001} & 0.038 & 32 & 0 \\
		Trailer & \textbf{0.244} & 0.856 & 1092 & 86 \\
		Utah Teapot & \textbf{0.015} & 0.067 & 112 & 24 \\
		Vase & \textbf{0.007} & 0.069 & 112 & 4 \\
		\bottomrule
	\end{tabular}

	\label{tab:parametric_benchmark}
\end{table}
We have implemented and parallelized our method in C++ both on CPU and on GPU (OpenGL). Mesh results were evaluated on an AMD Ryzen 9 7950X3D 16-Core CPU with an NVIDIA\textsuperscript{\textregistered} GeForce RTX\textsuperscript{\texttrademark} 4090 GPU, parametric results were evaluated on an Intel\textsuperscript{\textregistered} Core\textsuperscript{\texttrademark} Ultra~9~285K CPU. 
On a CPU, we use Intel Embree \cite{Embree} for ray-mesh intersections. 
On a GPU, we implemented an intersector reusing rays across query points in a grid.
Using raytracing hardware would likely improve the performance further.

We have extensively validated our method's accuracy, performance, and robustness. We demonstrate our results on meshes and parametric surfaces of different complexity (Fig.~\ref{fig:teaser}, Fig.~\ref{fig:results_meshes}, Fig.~\ref{fig:results_parametric}), surfaces with multiple boundaries (Fig.~\ref{fig:teaser}, Fig.~\ref{fig:results_parametric}), polygon soups that technically have numerous boundaries (Fig.~\ref{fig:results_soups}), surfaces of different genus and connected components (Fig.~\ref{fig:results_meshes}). We mostly demonstrate results on meshes with boundary; without boundary our method is a trivial inside-outside test (Fig.~\ref{fig:results_soups}, left). 

\subsection{Performance} 
\label{sec:results_performance}
\paragraph{Meshes} For a mesh with $F$ faces, performing a ray-intersection query using a well-balanced spatial hierarchy is $\mathcal{O}(\log F)$ on average, evaluating the boundary integral for $B$ boundary segments is $\mathcal{O}(B)$. In total, the time complexity of our algorithm for a single query point is $\mathcal{O}(B + \log F)$. For a regular grid of query points, e.g., in a voxelization of $W \times H \times D$ voxels, it takes $\bigO (B+W H (\log F + D))$ in total. The bottleneck of our method is computing the sum of areas, taking almost $92\%$ of the total time.

\begin{center}
	\includegraphics[width=\columnwidth]{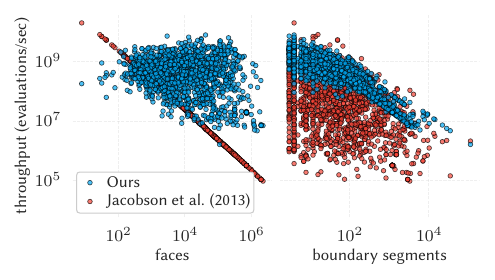}
	\captionof{figure}{
		Performance comparison (GPU) for meshes: Compared to our parallelization of the na\"ive method in \cite{Jacobson13winding}, our method is significantly faster while maintaining the same accuracy.
	}
	\label{fig:performance_boundary_gpu}
	\label{fig:performance_faces_gpu}
\end{center}

We compare our method's performance with the hierarchical winding number algorithm \cite{Jacobson13winding}, Fast Winding Numbers \cite{Barill2018FW}, and the One-Shot method \cite{Martens2025WindingNumberOneShot}. For closed shapes, both our method and One-Shot become a trivial inside-outside test, significantly faster than the other methods. We have run a comparison over 700 randomly selected closed meshes from \citet{Thingi10K}, our method is $73\times-\:127\times$ faster than the hierarchical method of \citet{Jacobson13winding}, $27\times-\:38\times$ faster than \citet{Barill2018FW}. For the rest of the comparisons we focus on the \emph{open} surfaces in \cite{Thingi10K}, 1417 meshes in total. For \cite{Jacobson13winding} and \cite{Barill2018FW} we used the libigl implementations \cite{libigl}. The libigl \verb|parallel for| implementation is  inefficient
; we have replaced it with our implementation that speeds up both their methods and ours. 
For all methods, we do not include precomputation time in the measurements (e.g., hierarchy construction in \cite{Jacobson13winding}), even though they can be substantial; our method's precomputations are significantly faster (Supp.~Fig.~\ref{fig:precomp_scatter}).
For each mesh, we sample $32^3$ query points in a regular grid of the mesh's bounding box. To make the comparison more fair, we do not leverage the regularity of query points, so we do not reuse the rays (Sec.~\ref{sec:mesh_intersect}); with ray reuse, our method is even faster (Supp.~Sec.\ref{sec:ray_reuse}). The results are in Fig.~\ref{fig:performance_boundary}; the speedup factors are in Fig.~\ref{fig:speedup_vs_alec}. 

\begin{center}
	\includegraphics[width=\columnwidth]{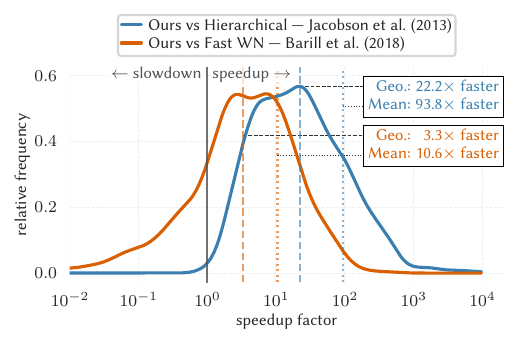}
	\captionof{figure}{
		Performance comparison (CPU) on meshes: 
		Histograms of the speedup factors (log) of our method compared to \cite{Jacobson13winding} ({\color[HTML]{3A7FB0}blue}) and \cite{Barill2018FW} ({\color[HTML]{D95F02}orange}).
		Depending on the summary statistics, we see $22\times$ to $93.8\times$ speedup ({\color[HTML]{3A7FB0}blue}) and $3.3\times$ to $10.6\times$ ({\color[HTML]{D95F02}orange}).
	}
	\label{fig:speedup_vs_alec}
	\label{fig:speedup_vs_fastwn}
	\label{fig:speed_up_cpu_alec}
	\label{fig:speed_up_cpu_fwn}
\end{center}

On a multi-threaded CPU, our method is on average $22\times-\;94\times$ faster than the hierarchical method of \cite{Jacobson13winding} (geometric vs. arithmetic mean of the speedup factors). Our method is also strictly faster on almost all the meshes, except for the rare case of polygon soups, where the number of boundary segments is $\Theta(F)$, so the methods' complexities are similar. In general, even though the method of \cite{Jacobson13winding} has empirical time complexity of $\mathcal{O}(F^{0.55})$ and for regularly tessellated meshes $B \approx \Theta(\sqrt{F})$, so one might think that the two methods will have similar performance at least in the asymptotic sense, we can see that in practice the difference is significant. This phenomenon can be explained by the distribution of the ratio of number of faces and the number of boundary segments in the dataset (Supp. Fig.~\ref{fig:dataset_faces_vs_boundaries}).

\begin{center}
	\includegraphics[width=\columnwidth]{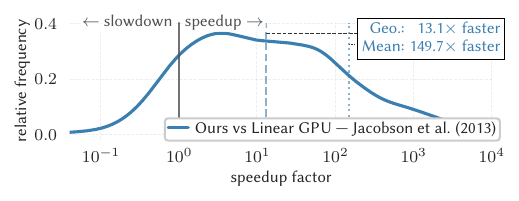}
	\captionof{figure}{
		Performance comparison (GPU) on meshes: The histogram of the speedup factors (log) of our method compared to our parallelization of \cite{Jacobson13winding}. 
		Depending on which summary statistics to use, we see $13\times$ to $149\times$ speedup on the data set.
	}
	\label{fig:speedup_vs_alec_gpu}
\end{center}

Compared with the fast approximation method of \citet{Barill2018FW}, our method is $3.3\times-\;10.6\times$ faster (geometric vs. arithmetic mean). While theoretically their method has a lower time complexity of $\bigO(\log F)$, in practice their method is consistently faster only for very large meshes (>100k faces). Importantly, their method is an approximation, while our method is precise to arbitrary accuracy. We show the comparison with order 2 approximation only; lower-order approximations are faster, but very imprecise (Sec.~\ref{sec:accuracy}, Supp.).

We have also compared our GPU performance with the parallelization of the na\"ive version of \citet{Jacobson13winding}.  Our method is $13\times-\;149\times$ faster (geometric vs. arithmetic mean) and our throughput reaches $10^9$ queries per second for moderately complex meshes, which is equivalent to computing a 4K image at 120 FPS (Fig.~\ref{fig:speedup_vs_alec_gpu}).

Finally, compared with \citet{Martens2025WindingNumberOneShot}, our method is strictly faster: their sphere arrangement is the bottleneck, and their complexity grows exponentially with the number of boundary components. Our method scales linearly with the number of boundary edges and does not depend on the number of boundary components (Supp. Fig.\ref{fig:boundary_processing_performance}).

\paragraph{Parametric Surfaces} We compare the performance of our algorithm with the concurrent method of \cite{Spainhour2026}, using the authors' implementation and their dataset \cite{llnl_axom_data}, as well as with the surface adaptive quadrature (AQ) method. 
To compute line integrals, we use our implementation of 15-point Gauss-Kronrod integrator with a tolerance of $\tau = 10^{-6}$. Following \citet{Spainhour2026}, as adaptive Gauss-Kronrod samples the integrand at predictable parameter values, we cache those function evaluation results for efficiency.

The formulation in \citet{Spainhour2026} shares a similar structure with ours: in both cases, a GWN is a sum of a ray-surface intersection term and a line integral, but the two expressions are different. More importantly, however, their method treats \textit{far-field} query points differently from \textit{near-field} ones, with the latter taking more time; they process all input patches separately. In contrast, we process all query points in the same way and instead of integrating over each patch boundary, we only take the \textit{surface} boundary; this critical difference makes our method significantly more efficient.

We conservatively identify two patch boundary curves $f(t)$, $g(t)$, $t\in [0,1]$ if $| |f-g| |_{L^\infty} < 10^{-8}$. The patch boundary curves that do not get identified with other curves form the surface boundary. While only the surface boundary is integrated, all patch boundary curves are considered for the ray-intersection trim test. 

We compare the performance of our method with the one of \citet{Spainhour2026} for the query points in a voxel grid of $32^3$ resolution (Table~\ref{tab:parametric_benchmark}). Our method is $5.6\times-\:16\times$ faster than theirs on average (geometric vs. arithmetic mean), while being noticeably slower only on one model out of 18 in their dataset. Identifying surface boundaries grants our method a $2.47 \times$ speedup on the dataset.

We also compare our method with Eq.~\ref{eq:wn_solid_angle} integrated via surface AQ with the 15-point Gauss-Kronrod rule. We tested a variety of tolerance values, such as $\tau = 10^{-k}, k \in \{5, 6, 7, 8\}$ on a benchmark of 100 cubic B\'ezier patches with control points randomly generated in the unit cube in a voxelization scenario. On average, our method takes $0.02$ ms while AQ takes $2.26$ ms per query point.

\subsection{Accuracy and Robustness}
\label{sec:accuracy}

\paragraph{Meshes} To analyze accuracy, for each mesh we sample 100 uniformly distributed random points in the bounding box, and compute the ground truth values of the GWN by using \cite{Jacobson13winding} evaluated in double for each triangle, accumulated in a highly accurate double-double accumulator. We then compute the error distribution for \cite{Jacobson13winding} in double, \cite{Barill2018FW} in double with different orders (0,1,2), and our method in both double (\texttt{f64}, both with double accumulator and double-double) and 32-bit float (Fig.~\ref{fig:accuracy}). As the plot shows, our accuracy with \texttt{f64} is on par with both methods in \cite{Jacobson13winding}. Our \texttt{f32} accuracy is at the limit of the type precision. The approximation method \cite{Barill2018FW} is imprecise, especially for the orders 0 and 1. 

The roundoff error is not a major concern for our algorithm: As the difference between our two \texttt{f64} curves show, the roundoff error due to accumulation of small terms is not an issue. The spherical triangle area depends on the \texttt{atan2} function, which is almost correctly rounded in modern implementations \cite{gladman:hal-03141101}.

Fig.~\ref{fig:accuracy} also suggests that the edge cases (e.g., collinear/overlapping boundary segments that are present in $30\%$ of our meshes) are handled correctly. Despite the requirement of transversality for the proofs, in practice it is not a restriction.

The analysis of our method with respect to off-by-integer errors is in the Supp. Sec.~\ref{sec:mesh_robustness}. In a nutshell, we can eliminate all such errors arbitrarily close to the surface at virtually no computational cost.
\begin{figure}[t]
	\centering
	\includegraphics[width=\columnwidth]{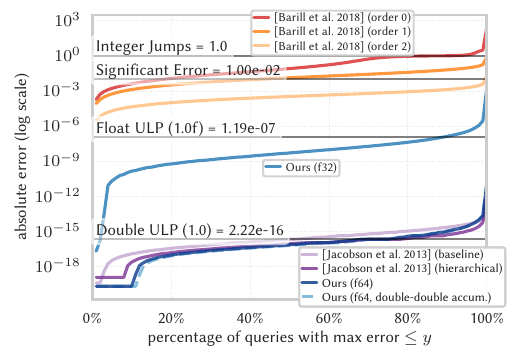}
	\captionof{figure}{
		Accuracy comparison: 
		how many queries in \% are below the plotted absolute error (log). 
		Our method in \texttt{f64} is as precise as the method by \cite{Jacobson13winding}. 
		The approximation method of \cite{Barill2018FW} is very imprecise; 
		orders 0 and 1 can be considered not useful for applications.
		The horizontal lines show important markers: 
		ULP(1.0) is the round-off error for winding number 1. 
	}
	\label{fig:accuracy}
\end{figure}

\paragraph{Parametric Surfaces} 

We validate accuracy of our method on parametric surfaces by comparing it with the state-of-the-art method for parametric surfaces \cite{Spainhour2026} as well as with the GWN of triangulations of the parametric surface.

We compare with \citet{Spainhour2026} on their whole parametric dataset, which includes high-order (>8) curves, over a $50^3$ regular grid of query points. We ignore the query points on the surface (distance < $10^{-6}$), as the GWN there is undefined. For all query points when our results and their results disagree (i.e., the absolute difference is $>10^{-4}$), we compare which method gives closer results to a GWN of high-resolution triangulation (linear deflection of $10^{-5}$). Their method relies on the disk radius parameter. For the default value ($10^{-2}$), their method has high errors compared to a high-resolution triangulation and our method on a few models (0.24 on model ‘Joint’ and 0.1 on ‘Boxed Sphere’). To correct their error, we manually adjusted their radius on these models to $10^{-4}$ and $10^{-3}$ respectively. The final RMSE difference between our results and theirs on all shapes is $8 \cdot 10^{-6}$, max $2.4  \cdot  10^{-3}$, with  99.92\% points differing by less than $10^{-4}$. Our methods differ on some rare points where quadrature is imprecise, either for their method or ours. Over those points, compared to the high-resolution triangulation, \citet{Spainhour2026} differs by $3.6 \cdot 10^{-4}$ (RMSE) and $2.4 \cdot 10^{-3}$ (max); Ours $6.8 \cdot 10^{-4}$ (RMSE) and $1 \cdot 10^{-3}$ (max). Note that while we use the same quadrature tolerance $(10^{-6})$ and maximum depth $(25)$ for both methods, the quadrature schemes are different (we use 15-point Gauss-Kronrod, they use 15-point Gauss-Legendre). Overall, the two methods are comparable in accuracy. 

We also compare the accuracy of our method on parametric surfaces with \cite{Jacobson13winding} on different meshings of the parametric surfaces, as well as with the adaptive quadrature over the surface, integrating Eq.~\ref{eq:wn_solid_angle}. For the meshing experiments, we used nine parametric models from the dataset \cite{llnl_axom_data} with linear deflection of $10^{-k}, k\in\{1,\ldots,5\}$. The error plots are in the supplementary (Supp.~Fig.~\ref{fig:mesh_vs_parametric_accuracy}). Meshing introduces the error around curved surfaces, propagating to the errors near the surface in the GWN field. With the mesh discretization while the applications control the level of tolerated accuracy, knowing the level of discretization necessary to guarantee the winding number accuracy might be difficult if the query points are not known \textit{a priori}. 

To compare accuracy of our method with a surface AQ, we use a random NURBS patch. Supp.~Fig.~\ref{fig:surface_AQ} shows we get higher accuracy while using significantly fewer surface evaluations.  Most of our surface evaluations are performed whenever $x_0$ is near the projected surface boundary. 

%
%
\section{Conclusions and Future Work}
We show that the computation of a winding number is equivalent to a ray-surface intersection and a line integral that includes the index of the projected surface boundary. This insight, upon discretization, leads to a simple, elegant, and effective algorithm that is significantly faster than the state of the art while being precise and robust.

Future work may explore fast approximation of the fractional part of our formula by adaptively simplifying the surface boundary or extending our formulation to other representations like oriented point clouds.

\begin{acks}
We acknowledge the support of the Natural Sciences and Engineering Research Council of Canada (NSERC) under Grant No.: RGPIN-2024-04968 (``Modelling and animation via intuitive input''), the NSERC - Fonds de recherche du Québec - Nature et technologies (FRQNT) NOVA Grant No. 314090, FRQNT team grant No. 361570, and a gift from Adobe. We thank Pierre Poulin for carefully proofreading an early draft of the manuscript and identifying several issues. The teaser models are a modified clock (TheGoofy, Thing ID: 328569) and a bobbin (ABC 933).
\end{acks}

\bibliographystyle{ACM-Reference-Format}
\bibliography{wn}
\clearpage

\begin{center}
	\includegraphics[height=0.90\textheight]{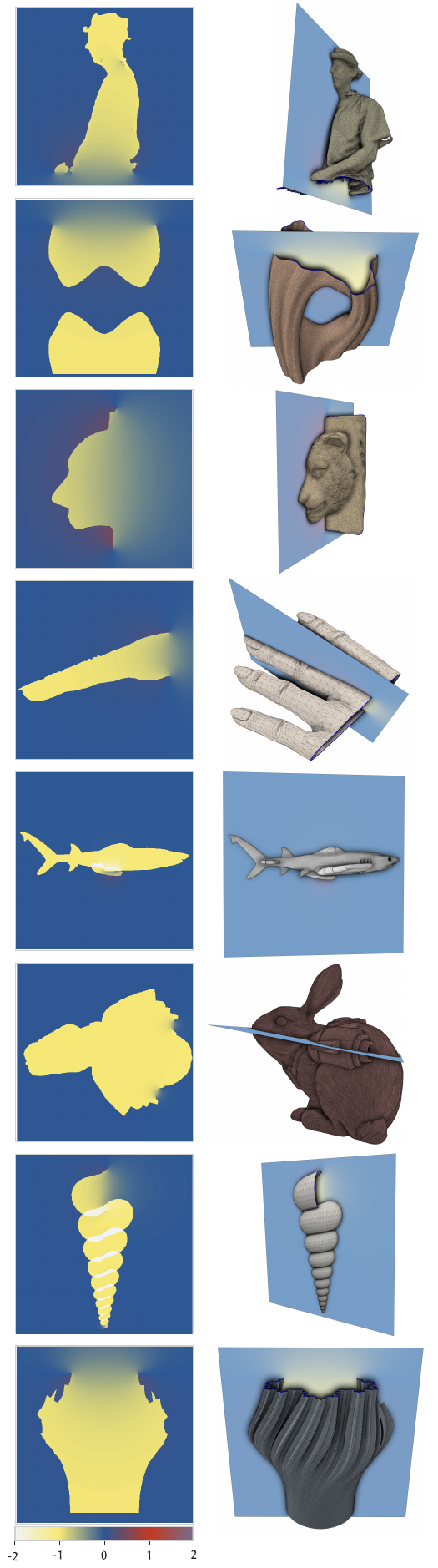}
	\captionof{figure}{
		A gallery of results on meshes. The boundary is highlighted in blue.  Model credits (Author: Thing ID): techknight (28601), virtox (28123), BrianStamile (17342), owenscenic (11618), gpvillamil (124774), Makerbot (226753), TeamTeamUSA (13668), hakalan (104694).
		\label{fig:results_meshes}
	}
\end{center}

\begin{center}
	\includegraphics[height=0.90\textheight]{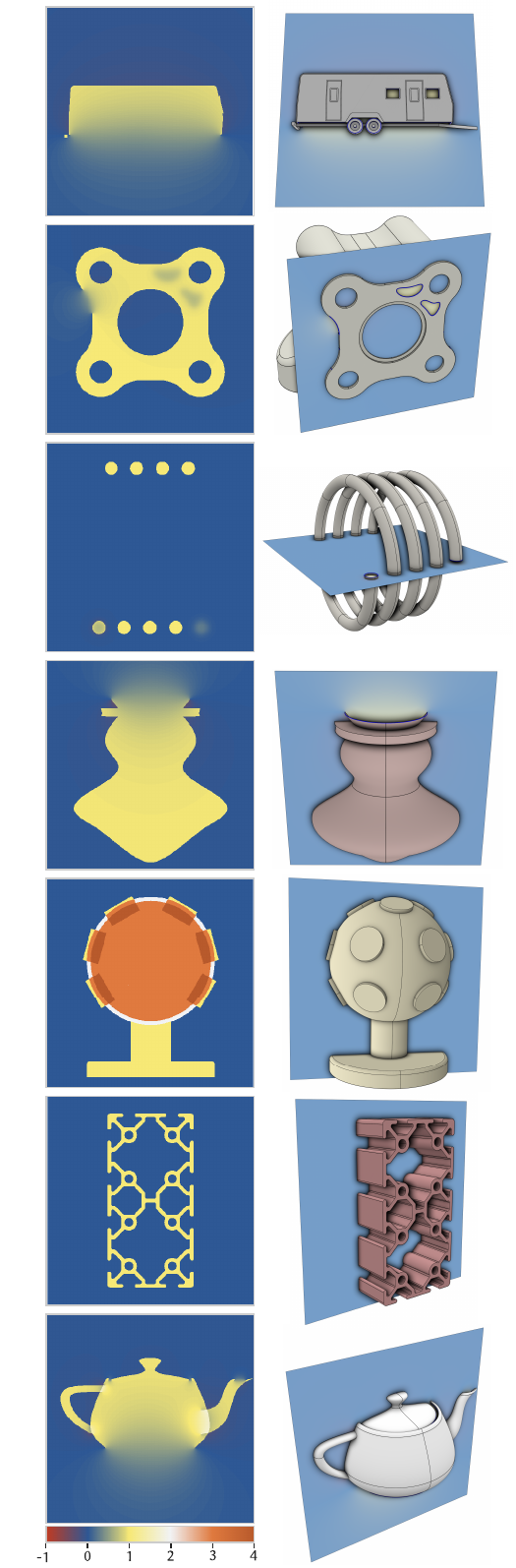}
	\captionof{figure}{
		A gallery of results on parametric surfaces, including trimmed NURBS (top to bottom: a trailer (ABC 4192), a joint (ABC 13), a spring (ABC 86), a vase \cite{Martens2025WindingNumberOneShot}, a lamp (ABC 3800), and a slide (ABC 4237)) and B'ezier (the Utah teapot \cite{newell1975utah_teapot}). Models are from the ABC dataset \cite{Koch_2019_CVPR} unless noted.\label{fig:results_parametric}
	}
\end{center}


\setcounter{section}{0}
\setcounter{figure}{0}
\setcounter{table}{0}
\setcounter{equation}{0}
\setcounter{algocf}{0}

\captionsetup{hypcap=false}
\renewcommand{\theHsection}{supp.\arabic{section}}
\renewcommand{\theHsubsection}{supp.\arabic{section}.\arabic{subsection}}
\renewcommand{\theHsubsubsection}{supp.\arabic{section}.\arabic{subsection}.\arabic{subsubsection}}
\renewcommand{\theHfigure}{supp.\arabic{figure}}
\renewcommand{\theHtable}{supp.\arabic{table}}
\renewcommand{\theHequation}{supp.\arabic{equation}}
\renewcommand{\theHtheorem}{supp.\arabic{section}.\arabic{theorem}}
\makeatletter
\patchcmd{\algocf@caption@algo}
  {\renewcommand{\theHalgocf}{\thealgocf}}
  {\renewcommand{\theHalgocf}{supp.\thealgocf}}
  {}{}
\makeatother

\twocolumn[{%
	\begin{center}
		\vspace{0.5em}
		{\LARGE\bfseries Supplementary Material}\\[0.4em]
		{\large for ``The Antipodal Method: Fast, Accurate, and Robust 3D Generalized Winding Numbers''}
		\vspace{1em}
	\end{center}
}]

\section{Proof of Lemma 4.1}
\begin{proof}
	For each region $\Omega_i$ with boundary $\partial \Omega_i$, we can express the region's area $A_i$ via Gauss-Bonnet theorem on the unit sphere, where the Gaussian curvature $K=1$:

	\begin{equation}
		A_i = \int_{\Omega_i} K \; dA = 2\pi \; \EulerChar_i - \int_{\partial \Omega_i} k_g \; ds - \sum_{j \in \mathcal{V}(\Omega_i)} \theta_j,
	\end{equation}
	where $\EulerChar_i$ is the Euler characteristic and $\theta_j$ are the turning angles of $\Omega_i$ formed by intersections (Fig.~\ref{fig:proof}a).

	\begin{figure}
		\includegraphics[width=\linewidth]{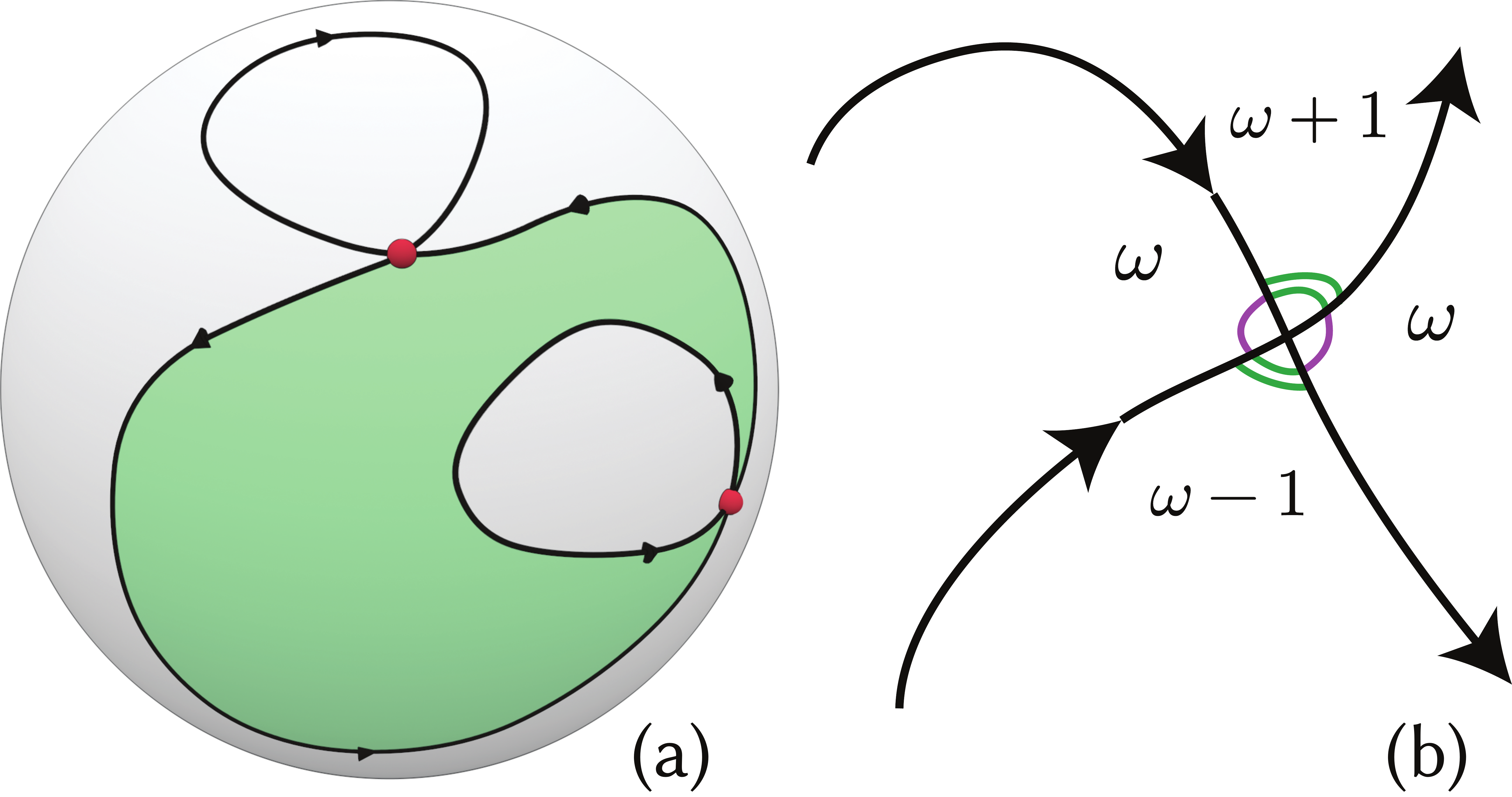}
		\caption{A closed curve $\Curve$ with only transversal self-intersections divides the sphere into regions of Euler characteristic of 1 (a). Around a self-intersection, the spherical winding number follows a particular pattern (b).}
		\label{fig:proof}
	\end{figure}

	Note that each region on the sphere is homeomorphic to an open disk, so all $\EulerChar_i =1$, even if the boundary of those regions has repeated vertices (Fig.~\ref{fig:proof}a), as we are only interested in the open regions (\cite{10.5555/29358}, \S1.4.3). 
	Multiplying each equation by the corresponding value of its spherical winding number $\Weight_i$ and summing, we get:
	\[
	\sum_{i=1}^n \Weight_i A_i = 2\pi \sum_{i=1}^n \Weight_i - \underbrace{\sum_{i=1}^n \Weight_i \int_{\partial \Omega_i} \; k_g \; ds}_{\text{\clap{\encircle{A}}}} - \underbrace{\sum_{i=1}^{n} \sum_{j \in \mathcal{V}(\Omega_i)} \Weight_i \theta_j}_{\text{\clap{\encircle{B}}}}.
	\]
	The second term $\encircle{A}$ traverses each curve segment of $\Curve$ once in each direction. Since every segment is shared between two adjacent regions, say $\Omega_1, \Omega_2$, with $\Omega_1$ lying on the left side of the curve with respect to its orientation, then $\Weight_{2} = \Weight_{1} + 1$ due to the Alexander numbering. Then, since reversing the boundary of the integral negates the integral, the second term $\encircle{A}$ simplifies to $\int_\Curve k_g \; ds$.

	To simplify $\encircle{B}$, we assemble angles around each curve's self-intersection. For the $k^\mathrm{th}$ self-intersection, $k \in \Intersections$, we denote as $\Weight_k$ the spherical winding number of the region bounded by the two curves in the counterclockwise direction. Due to the Alexander numbering, the other $\Weight$ values follow the pattern in Fig.~\ref{fig:proof}b.

	Considering equal pairs of angles in Fig.~\ref{fig:proof}b, at $k^\textrm{th}$ intersection the sum of the adjacent angles is equal to $\frac{2\pi}{4} \sum_{j=1}^4 \Weight_{kj}$. Hence,

	\begin{equation}
		\sum_{i=1}^n \Weight_i A_i = 2\pi \left(\sum_{i=1}^n \Weight_i - \frac{1}{4}\sum_{k\in \Intersections} \sum_{j=1}^4 \Weight_{kj}\right) - \int_{\Curve} k_g ds.
	\end{equation}

	To complete the proof we will show by induction that

	\begin{equation}
		\ind(\Curve, \VF) = \sum_{i=1}^n \Weight_i -  \frac{1}{4}\sum_{k \in \Intersections} \sum_{j=1}^4 \Weight_{kj}.
		\label{eq:index_proof}
	\end{equation}

	For a simple curve, this is trivially true: The vector field has no singularities inside the curve (by our assumption, singularities are in the region with $\Weight_i=0$), so $\ind(\Curve,\VF)=\pm 1$ with the sign depending on the curve orientation (c.f. Lemma 2.6 in \cite{Chillingworth1972}). This is equal to $\sum_{i=1}^n \Weight_i$, where there are exactly two regions, so $n=2$, one $\Weight_i=\pm 1$, the other one is zero, and $\frac{1}{4}\sum_k^{|\Intersections|} \sum_j^4 \Weight_{kj} = 0$.

	Let us assume the statement is true for curves with $m-1 \geq 0$ self-intersections. Given a curve with $m > 0$ self-intersections, it is not simple, so it contains a null-homotopic loop $\Curve_0$ that contains nothing else, with a complement $\Curve_{m-1}$ (c.f. Lemma 2.1, \cite{Chillingworth1972}). The index is linear with respect to the curve concatenation, so $\textrm{ind}(\Curve,\VF) = \textrm{ind}(\Curve_0,\VF) + \textrm{ind}(\Curve_{m-1},\VF)$ and $\textrm{ind}(\Curve_0,\VF) = \pm 1$. By induction, Eq.~\ref{eq:index_proof} is true for $\Curve_{m-1}$. At the same time, inserting $\Curve_0$ onto the sphere only changes the sum on the right-hand side by the same $\pm 1$, because the first sum increases either by $\Weight_k$ or $\Weight_k-2$ (Fig.~\ref{fig:proof}b), the only two regions with consistent boundary orientation; if we are inserting the top region or if we are adding the bottom region, respectively, and the second sum increases by $\frac{1}{4}\sum_j \Weight_{kj} = \Weight_k-1$, so the difference is $\pm 1$.
\end{proof}

\section{Robustness and Mesh Predicates}
\label{sec:mesh_robustness}
To analyze the robustness of Alg.~\ref{alg:wnr} in the main document, we note that there are a few potential sources of errors worth attention: the summations of small terms, the spherical area computation, the intersection number computation, and the singularity in the projection onto the unit sphere. The roundoff error is discussed in the main document (Sec.~\ref{sec:accuracy}).


Then, there are two potential sources of off-by-integer errors: The \texttt{atan2} function has a $2 \pi$ discontinuity and the signed intersection number depends on ray-triangle intersection, which is discontinuous at the triangle boundary and has numerical issues with coplanar rays. Mathematically, these discontinuities are coupled: As we vary the query point away from the surface, the winding number changes smoothly, so a discrete change in one cancels the discrete change in the other.

To make our computation unconditionally robust, we need to ensure that these discontinuities are also coupled numerically. To this end, we leverage symbolic perturbation \cite{SimulationOfSimplicity1990, SymbolicTreatmentOfGeometricDegeneracies, QualitativeSymbolicPerturbation} and exact integer predicates \cite{NEHRINGWIRXEL2021103015, Trettner2022}. Below we show that the predicate for ray-triangle intersection and \texttt{atan2} share exactly the same discontinuity (Suppl. Sec.~\ref{sec:exact_predicates}).

If the ray with the direction $d$ hits the boundary, we set $d$ to another arbitrary direction. To avoid dealing with the query points on the surface, we symbolically perturb them, as explained below.

\subsection{Robust Predicates for Meshes}
\label{sec:exact_predicates}
Winding numbers are not defined on the surface, so instead of using the real-valued query point $q = (x, y, z)$, we use the symbolically perturbed query point $q_\varepsilon = (x + \varepsilon_1, y + \varepsilon_2, z + \varepsilon_3)$ with $0 \ll \varepsilon_3 \ll \varepsilon_2 \ll \varepsilon_1 \ll v$ for any $v \in \R, v > 0$, where $\varepsilon_i$ are infinitesimal numbers.
This infinitesimal move of the query point does not change the real winding number in the volume and yields a consistently defined number whenever $q$ is on the surface.

To simplify the analysis, we choose the ray direction $d = (0, 0, 1)$, turning the ray-triangle intersection test (the computation of signed intersection number $\chi$ in Alg~\ref{alg:wnr}) into a point in 2D triangle test on the $xy$-plane. Let the projected triangle be defined by $p_i = (x_i, y_i)$ for $i \in \{0, 1, 2\}$ and consider an edge $e_{01}$ from $p_0$ to $p_1$.
We can then define the outward 2D normal as $n_{01} = (y_0 - y_1, x_1 - x_0)$.
To classify the query point $q_\varepsilon$ with respect to this edge, we compute $\text{sign}((q_\varepsilon - p_0) \cdot n_{01})$, which expands to $\text{sign}((x + \varepsilon_1 - x_0)(y_0 - y_1) + (y + \varepsilon_2 - y_0)(x_1 - x_0))$. A projected query point is inside the projected triangle if all edge classifications are negative.
The argument to the sign is linear in $\varepsilon_i$.
It is evaluated via coefficient signs in lexicographic order: we evaluate the real part $(x - x_0)(y_0 - y_1) + (y - y_0)(x_1 - x_0)$ first.
If it is zero, we check the sign of the $\varepsilon_1$ coefficient, which is $y_0 - y_1$.
If that is also zero, we check the $\varepsilon_2$ coefficient, $x_1 - x_0$.
If this is also zero, we have $y_0 = y_1$ and $x_0 = x_1$, a degenerate triangle that is safe to ignore because it cannot be hit. Note that $\varepsilon_3$ does not participate in the ``inside-the-triangle'' test, only for ``which-side-of-the-triangle'', i.e.\ the supporting plane test.

Now we analyze the discontinuity in \texttt{atan2} in the spherical triangle area computation in Alg.~\ref{alg:wnr}. The real part is the sign of a $2 \times 2$ determinant, which can be evaluated using exact predicates.
To avoid the floating points, we apply the integer predicate approach of \cite{NEHRINGWIRXEL2021103015, Trettner2022} and discretize the input into a fixed-point representation.
Our predicate is simple enough that a $30$-bit quantization can be realized using $64$-bit integer arithmetic at virtually no performance impact.
Even a $60$-bit quantization (using $128$-bit arithmetic for the real sign) is barely measurable due to the \texttt{atan2} overhead.

The discontinuity is at sign changes of the numerator $\text{sign}((-d) \cdot (a \times b))$, which expands to $\text{sign}((x_1 - x - \varepsilon_1)(y_0 - y - \varepsilon_2) - (y_1 - y - \varepsilon_2)(x_0 - x - \varepsilon_1))$.
The real, $\varepsilon_1$, and $\varepsilon_2$ coefficients are the same as for the ray-triangle predicate, ensuring that as long as these signs are evaluated consistently, both discontinuities will always coincide in the numerical evaluation as well, leading to a robust solution.

\begin{figure}
	\includegraphics[width=0.7\columnwidth]{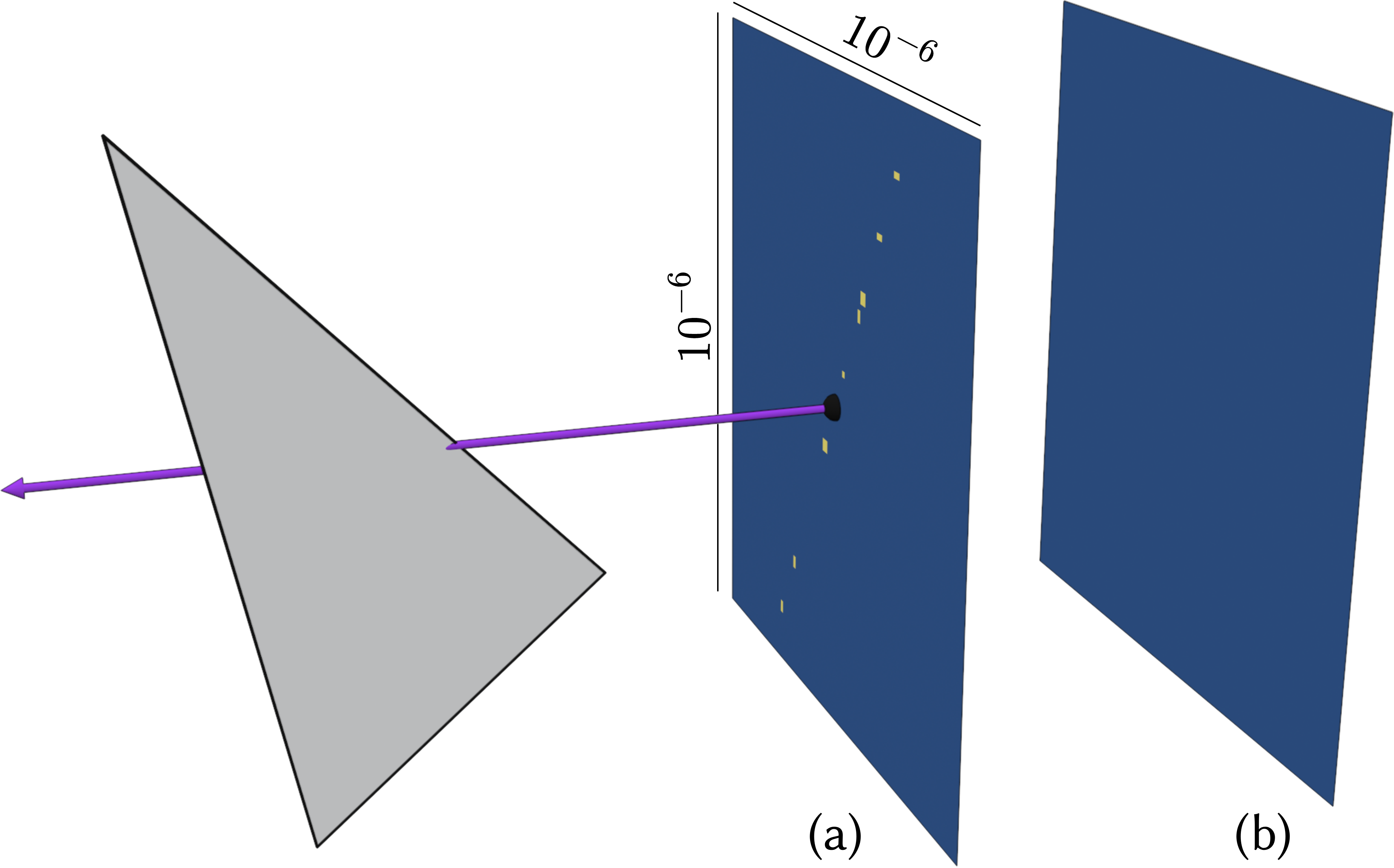}
	\captionof{figure}{
		Robustness test: We compute GWN for rays grazing the boundary of a triangle, where we move the query point in a small neighborhood (the square is a blow-up of a square with a side of $10^{-6}$). A na\"ive floating point implementation has integer jumps when the ray-intersection and the quadrant selection of the \texttt{atan2} are out of sync  (left, yellow dots). Our robust predicates guarantee no off-by-integer errors (right, no yellow dots).
	}
	\label{fig:results_robustness}
\end{figure}

\section{Ray-Mesh Intersections}
\label{sec:impl_details}
\label{sec:ray_reuse}

When query points lie on a regular grid, rays can be reused along each line of query points.
For a $W \times H \times D$ voxel grid ($W \leq H \leq D$), this reduces the number of ray-mesh intersections by a factor of up to $D$.

We simulated this optimization on our benchmark by scaling down the ray intersection cost accordingly and computed the resulting speedup over our baseline timings (Table~\ref{tab:ray_reuse}).
This is the same data set as in Fig.~\ref{fig:performance_boundary} of the main document.
Our timing consists of $t_\text{fractional} + t_\text{integer}$ where $t_\text{integer}$ is computed by ray-casting.
We added artificial methods $t_\text{fractional} + t_\text{integer} \cdot f$ for $f \in \{\frac{1}{32}, \frac{1}{1024}, 0\}$ to simulate perfectly implemented ray-use in a $32^3$ cube scenario, a $1024^2$ slice scenario, and the limit case of ``infinite ray use'', i.e.\ free raycasting.
This yields upper bounds for the obtainable speedups.

Even in the theoretical limit of free raycasting, the average speedup is only $1.38\times$ (geometric mean) or $1.86\times$ (arithmetic mean).
This is an instance of Amdahl's law: the fractional winding number evaluation dominates the total cost for most meshes.
The high maximum speedups occur for meshes with many faces but few boundary segments, where ray intersections constitute almost the entire computation.

\begin{table}[h]
	\centering
	\begin{tabular}{lrrrr}
		\toprule
		Scenario & Geo.\ Mean & Arith.\ Mean & Max \\
		\midrule
		$32^3$ voxel grid & $1.35\times$ & $1.51\times$ & $29.74\times$ \\
		$1024^2$ slice & $1.38\times$ & $1.77\times$ & $292.04\times$ \\
		Free raycasting & $1.38\times$ & $1.86\times$ & $408.16\times$ \\
		\bottomrule
	\end{tabular}
	\caption{Simulated speedup from ray reuse over our baseline method (no ray reuse). Even with zero ray intersection cost, the speedup is modest on average because the boundary integral dominates the runtime for most meshes.}
	\label{tab:ray_reuse}
\end{table}

\begin{figure*}[t]
	\includegraphics[width=\linewidth]{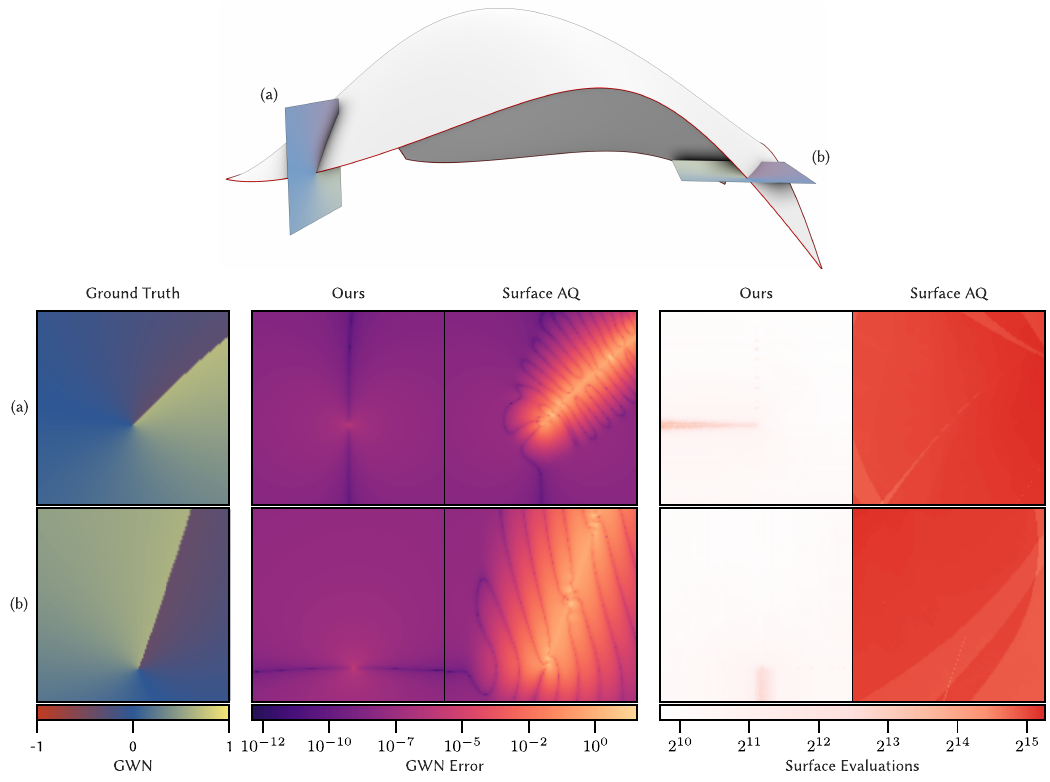}
	\caption{
		Compared with surface adaptive quadrature on a sample surface (above), our method is much more accurate while using orders of magnitude fewer surface evaluations, which explains our speedup (all plots are log scale). Note that our surface evaluations are concentrated at query points where the boundary projects near the singularities.
	}
	\label{fig:surface_AQ}
\end{figure*}

\section{Gradient of the GWN}

\begin{algorithm}[t]
		\caption{Gradient of GWN at point $\vec{p}$ for a mesh $\Mesh$ with boundary $\partial \Mesh$}
		\label{alg:dwnr}
		\KwIn{$\vec{p}$, $\Mesh$}
		$\vec{x_0} \gets \mathrm{RandomUnitVector()}$  \;
		$\vec{x_1} \gets -\vec{x_0}$\;
		$\nabla Area \gets \vec{0}$ \tcp*{Accumulates spherical area derivative}
		\ForEach{$(\vec{a},\vec{b}) \in \partial \Mesh\;$}{
			$\vec{r_0} \gets \vec{a} - \vec{p}$,\quad$\vec{r_1} \gets \vec{b} - \vec{p}$\;
			$l_0 \gets \lVert \vec{r_0} \rVert$,\quad$l_1 \gets \lVert \vec{r_1} \rVert$\;
			$\vec{v_0} \gets \frac{\vec{r_0}}{l_0}$,\quad$\vec{v_1} \gets \frac{\vec{r_1}}{l_1}$ \tcp*{Unit directions from $\vec{p}$}

			$n \gets \vec{x_1} \cdot (\vec{r_0} \times \vec{r_1})$ \tcp*{atan2 numerator}
			$d \gets l_0 l_1 + l_1 \vec{x_1} \cdot \vec{r_0} + l_0 \vec{x_1}\cdot \vec{r_1} + \vec{r_0}\cdot \vec{r_1}$ \tcp*{atan2 denominator}

			$c_n \gets \frac{2d}{(n^2 + d^2)}$ \tcp*{$2\,\partial \mathrm{atan2}(n,d)/\partial n$}
			$c_d \gets \frac{-2n }{(n^2 + d^2)}$ \tcp*{$2\,\partial \mathrm{atan2}(n,d)/\partial d$}

			$\vec{g_n} \gets -(\vec{r_1} \times \vec{x_1} + \vec{x_1} \times \vec{r_0})$ \tcp*{$\partial n / \partial \vec{p}$}
			$\vec{t_1} \gets l_1(\vec{x_1} + \vec{v_1} + (1 + \vec{x_1}\cdot\vec{v_1})\vec{v_0})$\;
			$\vec{t_2} \gets l_0(\vec{x_1} + \vec{v_0} + (1 + \vec{x_1}\cdot\vec{v_0})\vec{v_1})$\;
			$\vec{g_d} \gets -(\vec{t_1} + \vec{t_2})$ \tcp*{$\partial d / \partial \vec{p}$}

			$\nabla Area \gets \nabla Area + c_n \vec{g_n} + c_d \vec{g_d}$\;
		}
		\Return{$\dfrac{\nabla Area}{4\pi}$}
\end{algorithm}

	The winding number we compute is differentiable with respect to the query point almost everywhere. In particular, the derivative of the integer part of our formula is 0 almost everywhere. The last term of Eq.~\ref{eq:simpler_prop} smoothly depends on the surface boundary $\partial M$, i.e., its projection $\vec{\Gamma}$. Non-differentiability only arises when $\vec{\Gamma}$ passes through the chosen points $\South, \North$. We provide pseudocode for the gradient of our GWN for triangle meshes (See Supp. Algorithm~\ref{alg:dwnr}).

\section{Additional Comparisons}

We also compare the error distribution and the number of surface evaluations of surface adaptive quadrature with our method (see Sec.~\ref{sec:accuracy} in the main document) in Fig.~\ref{fig:surface_AQ}. We can see that our method has much higher accuracy using orders of magnitude fewer surface evaluations, which explains the time performance difference (our method is more than $110\times$ faster than surface adaptive quadrature).

\begin{figure}[t]
	\includegraphics[width=\linewidth]{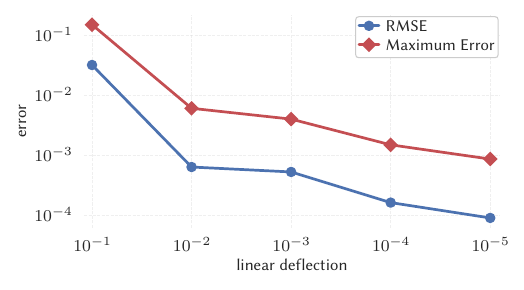}
	\caption{
		Comparing the accuracy of our method on parametric surfaces with \cite{Jacobson13winding} on meshes. We plot the average RMSE error as well as the maximum error of the GWN field over 9 parametric models from the dataset \cite{llnl_axom_data}, triangulated via \cite{OpenCASCADE2011} with different linear deflections.
	}
	\label{fig:mesh_vs_parametric_accuracy}
\end{figure}

\begin{figure}
	\includegraphics[width=\linewidth]{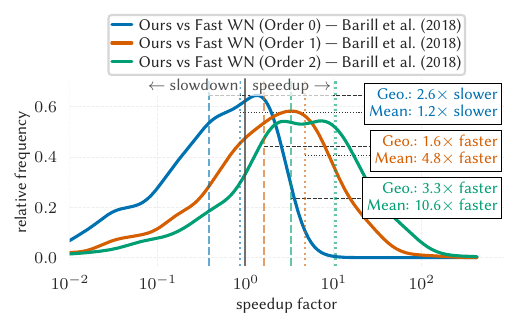}
	\caption{
		Histograms of the speedup factors (log) of our method compared to \cite{Barill2018FW} with different orders.
		This is the same data set as for Fig.~\ref{fig:speedup_vs_fastwn} of the main document.
		Our precise method is only slower than order 0, their roughest approximation, but faster than their orders 1 and 2.
		Note that the approximation errors of orders 0 and 1 are too high for practical use.
	}
	\label{fig:histograms_fw}
\end{figure}

\bibliographystyle{ACM-Reference-Format}
\bibliography{wn}

\begin{figure}
	\includegraphics[width=\linewidth]{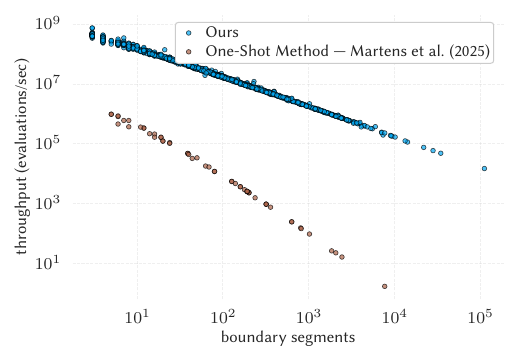}
	\caption{
		Our method is strictly faster than the one of \cite{Martens2025WindingNumberOneShot} (log-log plot, the higher the better). We test their method on the meshes from \cite{Thingi10K} with a single boundary component, their method does not support other cases. Their method requires computing an expensive spherical arrangement; instead, we only compute a line integral, explaining the orders of magnitude difference in performance.
	}
	\label{fig:boundary_processing_performance}
\end{figure}

\begin{figure}
	\includegraphics[width=\linewidth]{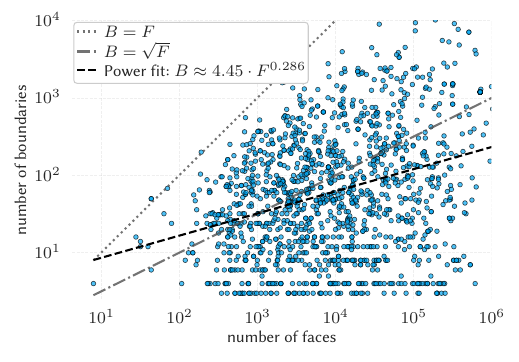}
	\caption{
		Log-log distribution of number of faces vs number of boundary segments in the open subset of the Thingi10K dataset.
		Theoretically, we would expect $O(\sqrt{F})$ many boundaries.
		In practice, we see $O(F^{0.286})$ on this data set but at the same time there's a significant amount of meshes with closer to $O(F)$ many.
	}
	\label{fig:dataset_faces_vs_boundaries}
\end{figure}

\begin{figure}[t]
	\includegraphics[width=\linewidth]{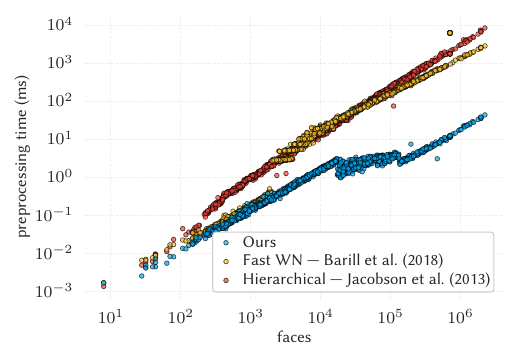}
	\caption{
		Preprocessing time (ms) vs.\ number of input faces on a log-log scale, comparing our method with \cite{Barill2018FW} and \cite{Jacobson13winding}.
			Both prior methods build custom acceleration structures whose construction cost grows substantially with mesh size.
			Our preprocessing is dominated by the Embree ray-intersection BVH, which stays below 100\,ms even for large meshes.
	}
	\label{fig:precomp_scatter}
\end{figure}

\end{document}